\documentclass[onecolumn,draftclsnofoot]{IEEEtran} 

\usepackage{etoolbox} 
\newtoggle{include_app} 
\toggletrue{include_app} 

\usepackage[utf8]{inputenc} 
\usepackage[T1]{fontenc} 
\usepackage[cmex10]{amsmath}
\usepackage[hidelinks]{hyperref}
\usepackage{url}
\usepackage{amsfonts,amssymb,latexsym,bm,amsthm}
\usepackage{nicefrac} 
\usepackage{microtype} 
\usepackage{xcolor}
\usepackage{graphicx}
\usepackage{rotating} 
\usepackage{cite}
\usepackage[normalem]{ulem} 
\usepackage{enumerate,subcaption}
\usepackage{algorithm,algpseudocode}
\usepackage{wrapfig}

\algnewcommand{\Define}[1]{%
  \State \textbf{define:}
  \Statex \hspace*{\algorithmicindent}\parbox[t]{.8\linewidth}{\raggedright #1}
}
\algnewcommand{\Inputs}[1]{%
  \State \textbf{inputs:}
  \Statex \hspace*{\algorithmicindent}\parbox[t]{.8\linewidth}{\raggedright #1}
}
\algnewcommand{\Initialize}[1]{%
  \State \textbf{initialize:}
  \Statex \hspace*{\algorithmicindent}\parbox[t]{.8\linewidth}{\raggedright #1}
}

\newtheorem{lemma}{Lemma}

\newcommand{\textb}[1]{\textcolor{black}{#1}}
\newcommand{\blue}{\color{black}}

\renewcommand{\hat}{\widehat}
\renewcommand{\bar}{\overline}

\renewcommand{\vec}[1]{\boldsymbol{#1}}
\newcommand{\ovec}[1]{\bar{\boldsymbol{#1}}}
\newcommand{\hvec}[1]{\hat{\boldsymbol{#1}}}

\newcommand{\defn}{\triangleq}
\newcommand{\mat}[1]{\ensuremath{\begin{bmatrix}#1\end{bmatrix}}}

 \newcommand{\mc}[1]{\ensuremath{\mathcal{#1}}}
\newcommand{\Real}{{\mathbb{R}}}
\newcommand{\Complex}{{\mathbb{C}}}

\newcommand{\tran}{^{\text{\textsf{T}}}}
\newcommand{\herm}{^{\text{\textsf{H}}}}
\newcommand*\dif{\mathop{}\!\mathrm{d}} 
\newcommand{\zero}{\mathbf{0}}

\DeclareMathOperator{\real}{Re}

\DeclareMathOperator{\E}{\mathbb{E}}
\DeclareMathOperator{\Exp}{\mathbb{E}}

\DeclareMathOperator{\Cov}{Cov}

\DeclareMathOperator{\tr}{tr}
\DeclareMathOperator{\diag}{diag}
\DeclareMathOperator{\gdiag}{gdiag}
\DeclareMathOperator{\Diag}{Diag}

\DeclareMathOperator{\prox}{prox}
\DeclareMathOperator{\gprox}{gprox}

\DeclareMathOperator{\abs}{abs}

\renewcommand{\eqref}[1]{(\ref{eq:#1})}

\newcommand{\Figref}[1]{Figure~\ref{fig:#1}}
\newcommand{\figref}[1]{Fig.~\ref{fig:#1}}
\newcommand{\tabref}[1]{Table~\ref{tab:#1}}

\newcommand{\secref}[1]{Sec.~\ref{sec:#1}}
\newcommand{\appref}[1]{Appendix~\ref{app:#1}}

\renewcommand{\algref}[1]{Alg.~\ref{alg:#1}}
\newcommand{\lineref}[1]{line~\ref{line:#1}}

\newcommand{\iter}{t}
\newcommand{\iters}{{t+1}}
\newcommand{\itero}{{t-1}}
\newcommand{\mmse}{_{\text{\sf mmse}}}

\newcommand{\true}{_0}
\newcommand{\full}{_{\text{\sf full}}}
\newcommand{\post}{p_{\text{x|y}}}
\newcommand{\prior}{p_{\text{x}}}
\newcommand{\like}{\ell}
\newcommand{\JGibbs}{J_{\text{\sf Gibbs}}}

\newcommand{\xbf}{\vec{x}}
\newcommand{\xbfhat}{\hvec{x}}
\newcommand{\rbf}{\vec{r}}
\newcommand{\ybf}{\vec{y}}

\newcommand{\ebf}{\vec{e}}
\newcommand{\wbf}{\vec{w}}
\newcommand{\fbf}{\vec{f}}
\newcommand{\Vbf}{\vec{V}}
\newcommand{\dbf}{\vec{d}}
\newcommand{\Ibf}{\vec{I}}
\newcommand{\ubf}{\vec{u}}

\newcommand{\Lambdabf}{\vec{\Lambda}}
\newcommand{\Sigmabf}{\vec{\Sigma}}
\newcommand{\Cbf}{\vec{C}}
\newcommand{\Abf}{\vec{A}}

\begin{document}
\setlength{\arraycolsep}{0.5mm}

\title{Denoising Generalized Expectation-Consistent Approximation for \textb{MR} Image Recovery}

\author{Saurav K. Shastri, Rizwan Ahmad, Christopher A. Metzler, and Philip Schniter%
\thanks{S. K. Shastri and P. Schniter are with the Dept. of Electrical and Computer Engineering, The Ohio State University, Columbus, OH 43201, USA, Email: \{shastri.19,schniter.1\}@osu.edu.}%
\thanks{R. Ahmad is with the Dept. of Biomedical Engineering, The Ohio State University, Columbus, OH 43201, USA, Email: ahmad.46@osu.edu}%
\thanks{C. A. Metzler is with the Dept. of Computer Science, The University of Maryland, College Park, MD 20742, USA, Email: metzler@umd.edu.}%
}

\maketitle


\begin{abstract}
To solve inverse problems, plug-and-play (PnP) methods replace the proximal step in a convex optimization algorithm with a call to an application-specific denoiser, often implemented using a deep neural network (DNN).
Although such methods \textb{yield accurate solutions}, they can be improved.
For example, denoisers are usually designed/trained to remove white Gaussian noise, but the denoiser input error in PnP algorithms is usually far from white or Gaussian. 
Approximate message passing (AMP) methods provide white and Gaussian denoiser input error, but only when the forward operator is \textb{sufficiently random}.
In this work, for Fourier-based forward operators, we propose a PnP algorithm based on generalized expectation-consistent (GEC) approximation---a close cousin of AMP---that offers predictable error statistics at each iteration, as well as a new DNN denoiser that leverages those statistics.
We apply our approach to magnetic resonance \textb{(MR)} image recovery and demonstrate its advantages over existing PnP and AMP methods. 
\end{abstract}


\section{Introduction}

When solving a linear inverse problem, we aim to recover a signal $\vec{x}\true\in\Complex^N$ from measurements $\vec{y}\in\Complex^P$ of the form
\begin{align}
\vec{y} = \vec{Ax}\true+\vec{w}
\label{eq:y},
\end{align}
where $\vec{A}$ is a known linear operator and $\vec{w}$ is unknown noise.
Well-known examples of linear inverse problems include
deblurring \cite{Wang:14};
super-resolution \cite{Park:SPM:03,Yang:TM:19};
inpainting \cite{Guillemot:SPM:13};
image recovery in magnetic resonance imaging (MRI) \cite{Knoll:SPM:20};
computed tomography \cite{Unser:FTSP:19};
holography \cite{Soulez:JOSAA:07};
and 
decoding in communications \cite{Venkataramanan:FTCIT:19}. 
Importantly, when $\vec{A}$ is not full column rank (e.g., when $P<N$), the measurements $\vec{y}$ can be explained well by many different hypotheses of $\vec{x}\true$.
In such cases, it is essential to harness prior knowledge of $\vec{x}\true$ when solving the inverse problem.

The traditional approach \cite{Fessler:SPM:20} to recovering $\vec{x}\true$ from $\vec{y}$ in \eqref{y} is to solve an optimization problem like 
\begin{align}
\hvec{x} 
&= \arg\min_{\vec{x}}
\big\{ g_1(\vec{x}) + g_2(\vec{x}) \big\}
\label{eq:opt},
\end{align}
where $g_1(\vec{x})$ promotes measurement fidelity and the regularization $g_2(\vec{x})$ encourages consistency with the prior information about $\vec{x}\true$.
For example, if $\vec{w}$ is white Gaussian noise (WGN) with precision (i.e., inverse variance) $\gamma_w$, then $g_1(\vec{x}) = \tfrac{\gamma_w}{2}\|\vec{y} - \vec{Ax}\|^2$ is an appropriate choice. 
Choosing a good regularizer $g_2$ is much more difficult.
A common choice is to construct $g_2$ so that $\vec{x}\true$ is sparse in some transform domain, i.e., $g_2(\vec{x}) = \lambda\|\vec{\Psi x}\|_1$ for $\lambda>0$ and a suitable linear operator $\vec{\Psi}$.
A famous example of this choice is total variation regularization \cite{Rudin:PhyD:92} and in particular its anisotropic variant (e.g., \cite{Shi:JAM:13}). 
However, the intricacies of many real-world signal classes (e.g., natural images) are not well captured by sparse models like these. 
Even so, these traditional methods provide useful building blocks for contemporary methods, as we describe below.
We will discuss the algorithmic aspects of solving \eqref{opt} in \secref{back}.

More recently, there has been a focus on training deep neural networks (DNNs) for image recovery given a sufficiently large set of examples $\{(\vec{x}_i,\vec{y}_i)\}$ to train those networks.
These DNN-based approaches come in many forms, including
dealiasing approaches \cite{Jin:TIP:17,Yang:TMI:17}, which use a convolutional DNN to recover $\vec{x}\true$ from $\vec{A}\herm\vec{y}$ or $\vec{A}^+\vec{y}$, where $(\cdot)^+$ denotes the pseudo-inverse;
unrolled approaches \cite{Hammernik:MRM:18,Monga:SPM:21}, which unroll the iterations of an optimization algorithm into a neural network and then learn the network parameters that yield the best result after a fixed number of iterations;
and inverse GAN approaches
\cite{Bora:ICML:17,Hand:COLT:18}, which first use a generative adversarial network (GAN) formulation to train a DNN to turn random code vectors $\vec{z}$ into realistic signal samples $\vec{x}$, and then search for the specific $\vec{z}$ that yields the $\hvec{x}$ for which $\|\vec{A}\hvec{x}-\vec{y}\|$ is minimal.
Good overviews of these methods can be found in \cite{Arridge:AN:19,Ongie:JSAIT:20,Hammernik:22}.
Although the aforementioned DNN-based methods have shown promise, they require large training datasets, which may be unavailable in some applications. 
Also, models trained under particular assumptions about $\vec{A}$ and/or statistics of $\vec{w}$ may not generalize well to test scenarios with different $\vec{A}$ and/or $\vec{w}$.

So-called ``plug-and-play'' (PnP) approaches \cite{Venkatakrishnan:GSIP:13} give a middle-ground between traditional algorithmic approaches and the DNN-based approaches discussed above.
In PnP, a DNN is first trained as a signal denoiser, and later that denoiser is used to replace the proximal step in an iterative optimization algorithm (see \secref{pnp}).
One advantage of this approach is that the denoiser can be trained with relatively few examples of $\{\vec{x}_i\}$ (e.g., using only signal patches rather than the full signal) and no examples of $\{\vec{y}_i\}$.
Also, because the denoiser is trained on signal examples alone, PnP methods have no trouble generalizing to an arbitrary $\vec{A}$ and/or $\vec{w}$ at test time.
The regularization-by-denoising (RED) \cite{Romano:JIS:17,Reehorst:TCI:19} framework yields a related class of algorithms with similar properties.
See \cite{Ahmad:SPM:20} for a comprehensive overview of PnP and RED.

With a well-designed DNN denoiser, PnP and RED significantly outperform sparsity-based approaches, as well as end-to-end DNNs in limited-data and mismatched-$\vec{A}$ scenarios (see, e.g., \cite{Ahmad:SPM:20}).
However, there is room for improvement.
For example, while the denoisers used in PnP and RED are typically trained to remove the effects of additive WGN (AWGN), PnP and RED algorithms
yield estimation errors that are not white nor Gaussian at each iteration.
As a result, AWGN-trained denoisers will be mismatched at every iteration, thus requiring more iterations and compromising performance at the fixed point.
Although recent work \cite{Gilton:TCI:21} has shown that deep equilibrium methods can be used to train the denoiser at the algorithm's fixed point, the denoiser may still remain mismatched for the many iterations that it takes to reach that fixed point, and the final design will be dependent on the $\vec{A}$ and noise statistics used during training.

These shortcomings of PnP algorithms motivate the following two questions: 
\begin{enumerate}
\item
Is it possible to construct a PnP-style algorithm that presents the denoiser with predictable error statistics at every iteration?
\item
Is it possible to construct a DNN denoiser that can efficiently leverage those error statistics?
\end{enumerate}
When $\vec{A}$ is a large unitarily invariant random matrix, the answers are well-known to be ``yes": approximate message passing (AMP) algorithms \cite{Donoho:PNAS:09} yield AWGN errors at each iteration with a known variance, which facilitates the use of WGN-trained DNN denoisers like DnCNN \cite{Zhang:TIP:17} (see \secref{pnp} for more on AMP algorithms).
In many inverse problems, however, $\vec{A}$ 
\textb{is either non-random or drawn from a distribution under which AMP algorithms do not behave as intended.}
So, the above two questions still stand.

In this paper, we answer both of the above questions in the affirmative for Fourier-based $\vec{A}$.
Using the framework of generalized expectation-consistent (GEC) approximation \cite{Fletcher:ISIT:16} in the wavelet domain \cite{Mallat:Book:08}, we propose a PnP algorithm that yields an AWGN error in each wavelet subband, with a predictable variance, at each iteration.
We then propose a new DNN denoiser design that can exploit knowledge of the wavelet-domain error spectrum. 
For recovery of MR images from the fastMRI \cite{Zbontar:18} and Stanford 2D FSE \cite{Ong:ISMRM:18} datasets, we present experimental results that show the advantages of our proposed approach over existing PnP and AMP-based approaches.
This paper builds on our recent conference publication \cite{Shastri:ICASSP:22} but adds our new denoiser design, much more background material and detailed explanations, and many new experimental results.


\section{Background} \label{sec:back}

\subsection{Magnetic resonance imaging} \label{sec:mri}

We now detail the version of the system model \eqref{y} that manifests in $C$-coil MRI.
There, $\vec{x}\true\in\Complex^N$ is a vectorized version of the $N$-pixel image that we wish to recover, $\vec{y}\in\Complex^{CM}$ are the so-called ``k-space'' measurements, and 
\begin{align}
\vec{A}=\mat{\vec{MF}\Diag(\vec{s}_1)\\[-2mm]\vdots\\[-2mm]\vec{MF}\Diag(\vec{s}_C)}
\label{eq:A}.
\end{align}
In \eqref{A},
$\vec{F}\in\Complex^{N\times N}$ is a unitary 2D discrete Fourier transform (DFT),
$\vec{M}\in\Real^{M\times N}$ is a sampling mask formed from $M$ rows of the identity matrix $\vec{I}\in\Real^{N\times N}$, and
$\vec{s}_c\in\Complex^{N}$ is the $c$th coil-sensitivity map. 
In the special case of single-coil MRI, we have $C=1$ and $\vec{s}_1=\vec{1}$, where $\vec{1}$ denotes the all-ones vector.
In MRI, the ratio $R\defn N/M$ is known as the ``acceleration rate.''
When $R>1$, the matrix $\vec{A}$ 
\textb{can be column-rank deficient and/or poorly conditioned}
even when $C\geq R$, and so prior knowledge of $\vec{x}\true$ must be exploited for accurate recovery.

\blue
In practical MRI, physical constraints govern the construction of the sampling mask $\vec{M}$.
For example, samples are always collected along lines or curves in k-space.
In clinical practice, it is most common to sample along lines parallel to one dimension of k-space, as illustrated in Figs.~\ref{fig:mask}(c)-(d) for 2D sampling. 
We will refer to this approach as ``2D line sampling.''
In this case, one dimension of k-space is fully sampled and the other dimension is subsampled.
For the subsampled dimension, it is common to sample pseudorandomly or randomly, but with a higher density near the k-space origin, as shown in Figs.~\ref{fig:mask}(c)-(d). 
Also, when using ESPIRiT to estimate the coil-sensitivity maps $\{\vec{s}_c\}$, one must include a fully-sampled ``autocalibration'' region centered at the origin, as shown in Figs.~\ref{fig:mask}(b)-(d).

2D line sampling, while attractive from an implementation standpoint, poses challenges for signal reconstruction due to high levels of coherence \cite{Foucart:Book:13} in the resulting $\vec{A}$ matrix. 
This has led some algorithm designers to consider ``2D point sampling'' masks such as those shown in \figref{mask}(a)-(b), since they yield $\vec{A}$ with much lower coherence \cite{Lustig:MRM:07}.
But such masks are rarely encountered in practical 2D MR imaging.
It is, however, possible to encounter a 2D point mask as a byproduct of the following 3D acquisition process: i) acquire a 3D k-space volume using 3D line sampling, ii) perform an inverse DFT along the fully sampled dimension, and iii) slice along that dimension to obtain a stack of 2D k-space acquisitions. 
The location of each line in 3D k-space determines the location of the respective point sample in 2D k-space, and these locations can be freely chosen. 
But 3D acquisition is uncommon because it is susceptible to motion; in 2D acquisition, the patient must lie still for the acquisition of a single slice, whereas in 3D acquisition the patient must lie still for the acquisition of an entire volume.
We include experiments with 2D point masks only to compare with the VDAMP family of algorithms \cite{Millard:OJSP:20,Metzler:ICASSP:21,Millard:22,Millard:ISMRM:22} discussed in the sequel, since these algorithms are all designed around the use of 2D point masks.
\color{black}

Although our paper focuses on MRI, the methods we propose apply to any application where the goal is to recover a signal from undersampled Fourier measurements. 

\begin{figure}
    \centering
    \includegraphics[width = \linewidth]{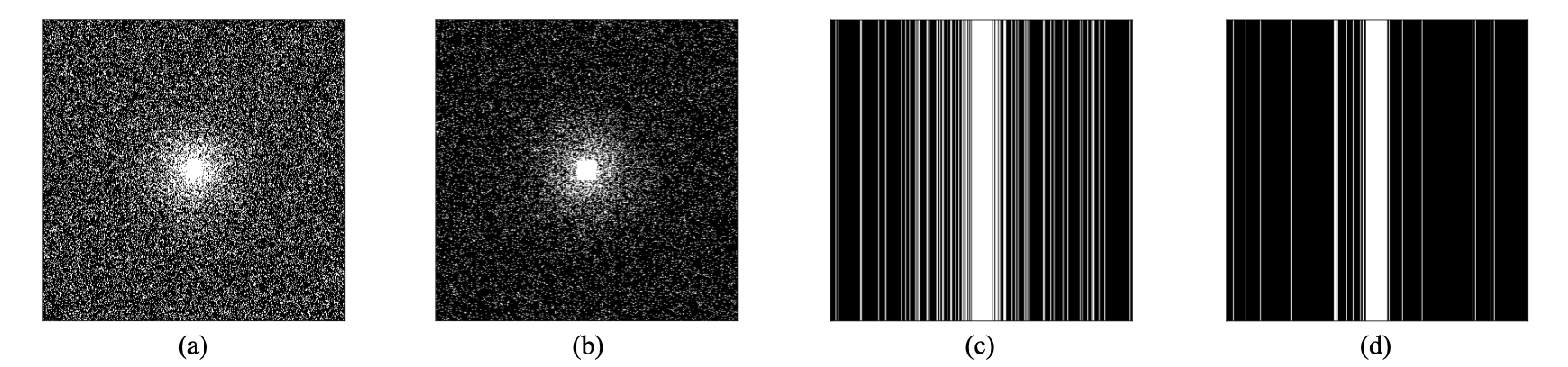}
    \caption{\textb{Examples of sampling masks $\vec{M}$: (a) 2D point sampling at $R=4$, (b) 2D point sampling at $R=8$ with a $24\times 24$ fully sampled central autocalibration region, (c) 2D line sampling at $R=4$ with a $24$-wide fully sampled central autocalibration region, and (d) 2D line sampling at $R=8$ with a $24$-wide fully sampled central autocalibration region}.}
    \label{fig:mask}
\end{figure}



\subsection{Plug-and-play recovery} \label{sec:pnp}


Many algorithms have been proposed to solve the optimization problem \eqref{opt} (see, e.g., \cite{Fessler:SPM:20}).
The typical assumptions are that $g_1$ is convex and differentiable, $\nabla g_1$ is Lipschitz with constant $L>0$, and $g_2$ is convex but possibly not differentiable, which allows sparsity-inducing regularizations like $g_2(\vec{x})=\lambda\|\vec{\Psi x}\|_1$.
One of the most popular approaches is ADMM \cite{Boyd:FTML:11}, summarized by the iterations 
\begin{subequations}\label{eq:admm}
\begin{align}
\vec{x}_1 &\gets \prox_{\gamma^{-1}g_1}(\vec{x}_2-\vec{u})
\qquad
\label{eq:admm_loss}\\
\vec{x}_2 &\gets \prox_{\gamma^{-1}g_2}(\vec{x}_1+\vec{u})
\label{eq:admm_prox}\\
\vec{u} &\gets \vec{u} + \left(\vec{x}_1 - \vec{x}_2\right)
\label{eq:admm_dual},
\end{align}
\end{subequations}
where $\gamma$ is a tunable parameter\footnote{
The parameter $\gamma$ arises from the augmented Lagrangian used by ADMM: $g_1(\vec{x}_1)+g_2(\vec{x}_2)+\real\{\vec{u}\herm(\vec{x}_1-\vec{x}_2)\} + \frac{\gamma}{2}\|\vec{x}_1-\vec{x}_2\|^2$.} 
that affects convergence speed but not the fixed point, and
\begin{align}
\prox_{\rho}(\vec{r})
\defn \arg\min_{\vec{x}} \big\{ \rho(\vec{x}) + \tfrac{1}{2}\|\vec{x}-\vec{r}\|^2 \big\}
\label{eq:prox}.
\end{align}
For example, when $g_1(\vec{x})=\frac{\gamma_w}{2}\|\vec{Ax}-\vec{y}\|^2$, we get
\begin{align}
\prox_{\gamma^{-1}g_1}(\vec{r})
= \big(\gamma_w\vec{A}\herm\vec{A}+\gamma\vec{I}\big)^{-1}\big(\gamma_w\vec{A}\herm\vec{y} + \gamma\vec{r}\big)
\label{eq:prox_L2}.
\end{align}
Based on the prox definition in \eqref{prox}, ADMM step \eqref{admm_prox}
can be interpreted as MAP estimation \cite{Poor:Book:94} of $\vec{x}\true$ with prior $p(\vec{x}\true)\propto e^{-g_2(\vec{x}\true)}$ from an observation $\vec{r}=\vec{x}\true+\vec{e}$ of the true signal corrupted by $\gamma$-precision AWGN $\vec{e}$, i.e., MAP denoising.
This observation led Venkatakrishnan et al.\ \cite{Venkatakrishnan:GSIP:13} to propose that the prox in \eqref{admm_prox}
be replaced by a high-performance image denoiser $\vec{f}_2:\Real^N\rightarrow\Real^N$ like BM3D \cite{Dabov:TIP:07}, giving rise to PnP-ADMM. 
It was later proposed to use a DNN-based denoiser in PnP \cite{Meinhardt:ICCV:17}, such as
DnCNN~\cite{Zhang:TIP:17}.
Note that when \eqref{admm_prox} is replaced with a denoising step of the form ``$\vec{x}_2\gets\vec{f}_2(\vec{x}_1+\vec{u})$,'' the parameter $\gamma$ does affect the fixed-point and thus must be tuned to obtain the best recovery accuracy.

The PnP framework was later extended to other algorithms, such as primal-dual splitting (PDS) in \cite{Ono:SPL:17,Meinhardt:ICCV:17} and proximal gradient descent (PGD) in \cite{Kamilov:SPL:17,Meinhardt:ICCV:17}.
For use in the sequel, we write the PGD algorithm as
\begin{subequations}\label{eq:pgd}
\begin{align}
\vec{x}_1 &\gets \vec{x}_2 - \mu \nabla g_1(\vec{x}_2)
\label{eq:pgd_grad}\\
\vec{x}_2 &\gets \prox_{\mu g_2}(\vec{x}_1)
\label{eq:pgd_prox},
\end{align}
\end{subequations}
where $\mu\in(0,1/L)$ and $L$ is the Lipschitz constant of $\nabla g_1$.
For example, when $g_1(\vec{x})=\frac{1}{2}\|\vec{Ax}-\vec{y}\|^2$, we get $\nabla g_1(\vec{x}) = \vec{A}\herm(\vec{Ax}-\vec{y})$.
For all of these PnP incarnations, the prox step in the original optimization algorithm is replaced by a high-performance denoiser $\vec{f}_2$.
As shown in the recent overview \cite{Ahmad:SPM:20}, PnP methods 
have been shown to significantly outperform sparsity-based approaches in MRI, 
as well as end-to-end DNNs in limited-data and mismatched-$\vec{A}$ scenarios.

Although PnP algorithms work well for MRI, there is room for improvement.  
For example, while image denoisers are typically designed/trained to remove the effects of AWGN, PnP algorithms do not provide the denoiser with an AWGN-corrupted input at each iteration.
Rather, the denoiser's input error has iteration-dependent statistics that are difficult to analyze or predict.

\subsection{Approximate message passing} \label{sec:amp}

For the model \eqref{y} with $\vec{w}\sim\mc{N}(0,\tau_w\vec{I})$, the AMP algorithm\footnote{
For generalized linear models, one would instead use the Generalized AMP algorithm from \cite{Rangan:ISIT:11}.}
\cite{Donoho:PNAS:09,Donoho:ITW:10a} manifests as the following iteration over $t=0,1,2,\dots$:%
\begin{subequations}\label{eq:amp}
\begin{eqnarray}
\vec{v}^\iters 
&=& \beta\cdot\big( \vec{y}-\vec{Ax}^\iter \big)
        + \tfrac{1}{M}\vec{v}^\iter \tr\{\nabla\vec{f}_2^\iter(\vec{x}^\itero\!+\!\beta\vec{A}\herm\vec{v}^\iter)\}
\qquad
\label{eq:amp_onsager}\\
\tau^\iters
&=& \tfrac{1}{M}\|\vec{v}^\iters\|^2
\label{eq:amp_noisevar}\\
\vec{x}^\iters
&=&\vec{f}_2^\iters(\vec{x}^\iter + \beta\vec{A}\herm\vec{v}^\iters)
\end{eqnarray}
\end{subequations}
initialized as $\vec{v}^0=\vec{0}=\vec{x}^0$, where 
$\vec{f}_2^\iter(\cdot)$ is the iteration-$t$ denoising function \textb{(which may depend on $\tau^t$)},
$\tr\{\nabla\vec{f}_2^\iter(\vec{r})\}$ is the trace of the Jacobian of $\vec{f}_2^\iter$ at $\vec{r}$,
and $\beta=\sqrt{N}/\|\vec{A}\|_F$.
The last term in \eqref{amp_onsager}, known as the ``Onsager correction,'' is a key component of the AMP algorithm.
Without it, \eqref{amp} would reduce to the PnP version of the PGD algorithm \eqref{pgd} with $\mu=\beta^2$.

The goal of Onsager correction is to make the denoiser input error
\begin{align}
\vec{e}^\iters
&\defn \vec{x}^\iter + \beta\vec{A}\herm\vec{v}^\iters - \vec{x}\true
\label{eq:e}
\end{align}
behave like a realization of WGN with variance $\tau^\iters$, where $\tau^\iters$ is given in \eqref{amp_noisevar}.
Note that if 
\begin{align}
\vec{e}^\iters
\sim\mc{N}(\vec{0},\tau^\iters\vec{I})
\label{eq:e_amp}
\end{align}
did hold, it would be straightforward to design the denoiser $\vec{f}_2^\iters$ for MAP or MMSE optimality.
For example, in \eqref{opt}, if we interpret $g_1(\vec{x})$ as the log-likelihood and $g_2(\vec{x})$ as the log-prior, then $g_1(\vec{x})+g_2(\vec{x})$ becomes the log-posterior (up to a constant) and so $\hvec{x}$ in \eqref{opt} becomes the MAP estimate \cite{Poor:Book:98}.
Thus, for the case of MAP estimation, we would use the MAP denoiser  $\vec{f}_2^\iter(\vec{r})=\prox_{\tau^\iter g_2}(\vec{r})$,
and $\vec{x}^\iter$ would approach the MAP estimate as $t\rightarrow\infty$ \cite{Donoho:PNAS:09}.
On the other hand, for the case of MMSE estimation, where we would like to compute the conditional mean $\hvec{x}\mmse \defn \E\{\vec{x}|\vec{y}\}$, we would use the MMSE denoiser $\vec{f}_2^\iter(\vec{r})=\E\{\vec{x}\,|\,\vec{r}\}$ for $\vec{r}=\vec{x}\true+\vec{e}$ with $\vec{e}\sim\mc{N}(\vec{0},\tau^\iter\vec{I})$ \cite{Donoho:ITW:10a}.

Importantly, when the forward operator $\vec{A}\in\Real^{P\times N}$ \textb{is} i.i.d.\ sub-Gaussian, the dimensions $P,N\rightarrow\infty$ with a fixed ratio $P/N$, and $\vec{f}_2^\iter$ is Lipschitz, \cite{Bayati:TIT:11,Berthier:II:19} established that the WGN property \eqref{e_amp} does indeed hold. 
Furthermore, defining the MSE $\mc{E}^\iter\defn \frac{1}{N}\|\vec{x}^\iter-\vec{x}\true\|^2$, \cite{Bayati:TIT:11,Berthier:II:19} established that AMP obeys the following scalar state-evolution over $t=0,1,2,\dots$:%
\begin{subequations}\label{eq:amp_se}
\begin{align}
\tau^\iter 
&= \tau_w + \tfrac{N}{P} \mc{E}^\iter 
\label{eq:ampse_tau}\\
\mc{E}^\iters
&= \tfrac{1}{N}\E\{\|\vec{f}_2^\iter(\vec{x}\true+\mc{N}(\vec{0},\tau^\iter\vec{I}))-\vec{x}\true\|^2\}
\label{eq:ampse_E}.
\end{align}
\end{subequations}
\textb{Remarkably, the AMP state evolution shows that, in the large-system limit, the trajectory of the mean-squared recovery error can be predicted in advance knowing only the dimensions of i.i.d.\ sub-Gaussian $\vec{A}$ (not the values in $\vec{A}$) and the MSE behavior of the denoiser $\vec{f}_2^t(\cdot)$ when faced with the task of removing white Gaussian noise.}
Moreover, when $\vec{f}_2^t$ is the MMSE denoiser and the state-evolution has a unique fixed point, \cite{Bayati:TIT:11,Berthier:II:19} established that AMP provably converges to the MMSE-optimal estimate $\hvec{x}\mmse$.
These theoretical results were first established for separable denoisers $\vec{f}_2$ in \cite{Bayati:TIT:11} and later extended to non-separable denoisers in \cite{Berthier:II:19}.
By ``separable'' we mean that $\vec{f}_2$ takes the form $\vec{f}_2(\vec{x})=[f_{2}(x_1),\dots,f_{2}(x_N)]\tran$ for some scalar denoiser $f_2:\Real\rightarrow\Real$. 

For practical image recovery problems, \cite{Metzler:ICIP:15} proposed to approximate the MMSE denoiser by a high-performance image denoiser like BM3D or a DNN, and called it ``denoising-AMP'' (D-AMP).
Since these image denoisers are non-separable and high-dimensional, the trace-Jacobian term in \eqref{amp_onsager} (known as the ``divergence'') is difficult to compute, and so D-AMP uses the Monte-Carlo approximation \cite{Ramani:TIP:08}
\begin{align}
\tr\{\nabla\vec{f}_2^\iter(\vec{r})\}
\approx \delta^{-1}\vec{q}\herm\big[\vec{f}_2^\iter(\vec{r}+\delta\vec{q})-\vec{f}_2^\iter(\vec{r)}\big]
\label{eq:trJfapprox} ,
\end{align}
where $\vec{q}$ is a fixed realization of $\mc{N}(\vec{0},\vec{I})$ and $\delta$ is a small positive number.
D-AMP performs very well with large i.i.d.\ sub-Gaussian $\vec{A}$, but can diverge with non-random $\vec{A}$, such as those encountered in MRI (recall \eqref{A}).

\subsection{Expectation-consistent approximation and VAMP} \label{sec:ec}

Expectation-consistent (EC) approximation \cite{Opper:NIPS:05} is an inference framework with close connections to both PnP-ADMM and AMP.
In EC, one is assumed to have access to the prior density $\prior(\vec{x})$ on $\vec{x}\true$ and the likelihood function $\like(\vec{x};\vec{y})$, and the goal is to approximate the mean of the posterior $\post(\vec{x}|\vec{y})$, i.e., the MMSE estimate $\hvec{x}\mmse=\E\{\vec{x}|\vec{y}\}$.
Although Bayes rule says that $\post(\vec{x}|\vec{y}) = Z^{-1}(\vec{y}) \prior(\vec{x})\like(\vec{x};\vec{y})$ for $Z(\vec{y})\defn\int \prior(\vec{x})\like(\vec{x};\vec{y}) \dif\vec{x}$, this integral is usually too difficult to compute in the high-dimensional case.
But note that we can write
\begin{align}
\post(\vec{x}|\vec{y})
&= \arg\min_{q} D\big(q(\vec{x})\big\|\post(\vec{x}|\vec{y})\big) \\
&= \arg\min_{q} 
    D\big(q(\vec{x})\big\|\like(\vec{x};\vec{y})\big)
    + D\big(q(\vec{x})\big\|\prior(\vec{x})\big)
    + H\big(q(\vec{x})\big) \\
&= \arg\min_{q_1,q_2,q_3} 
    \underbrace{
    D\big(q_1(\vec{x})\big\|\like(\vec{x};\vec{y})\big)
    + D\big(q_2(\vec{x})\big\|\prior(\vec{x})\big)
    + H\big(q_3(\vec{x})\big) 
    }_{\displaystyle \defn \JGibbs(q_1,q_2,q_3)}
    \text{~such that~} q_1=q_2=q_3
\label{eq:JGibbs},
\end{align}
where the minimizations are conducted over sets of probability densities,
$D(q_1\|\prior)\defn\int q_1(\vec{x})\log \frac{q_1(\vec{x})}{\prior(\vec{x})} \dif\vec{x}$ is the Kullback-Liebler (KL) divergence from $\prior$ to $q_1$, 
$H(q_3)\defn -\int q_3(\vec{x}) \log q_3(\vec{x}) \dif\vec{x}$ is the differential entropy of $q_3$,
and $\JGibbs(q,q,q)$ is known as the Gibbs free energy of $q$.
So, if \eqref{JGibbs} could be solved, it would give a way to compute the posterior that avoids computing $Z(\vec{y})$.
However, \eqref{JGibbs} is generally too difficult to solve, and so it was proposed in \cite{Opper:NIPS:05} to relax the equality constraints in \eqref{JGibbs} to moment-matching constraints, i.e.,
\begin{align}
\arg\min_{q_1,q_2,q_3} \JGibbs(q_1,q_2,q_3) 
    \text{~such that~} 
    \begin{cases}
    \E\{\vec{x}|q_1\} = \E\{\vec{x}|q_2\} = \E\{\vec{x}|q_3\} \\
    \tr(\Cov\{\vec{x}|q_1\}) = \tr(\Cov\{\vec{x}|q_2\}) = \tr(\Cov\{\vec{x}|q_3\}) ,
    \end{cases}
\label{eq:EC} 
\end{align}
where $\E\{\vec{x}|q_i\}$ and $\Cov\{\vec{x}|q_i\}$ denote the mean and covariance of $\vec{x}$ under $\vec{x}\sim q_i$ for $i=1,2,3$, respectively.
The authors of \cite{Opper:NIPS:05} then showed that the optimization problem \eqref{EC} is solved by the densities
\begin{align}
q_1(\vec{x};\vec{r}_1,\gamma_1) 
&\propto \like(\vec{x};\vec{y})\mc{N}(\vec{x};\vec{r}_1,\vec{I}/\gamma_1)\\
q_2(\vec{x};\vec{r}_2,\gamma_2) 
&\propto \prior(\vec{x})\mc{N}(\vec{r}_2;\vec{x},\vec{I}/\gamma_2)\\
q_3(\vec{x};\hvec{x},\eta)
&= \mc{N}(\vec{x};\hvec{x},\vec{I}/\eta)
\end{align}
for the values of $(\vec{r}_1,\gamma_1,\vec{r}_2,\gamma_2,\hvec{x},\eta)$ that lead to the satisfaction of the constraints in \eqref{EC}.
The resulting $\hvec{x}$ approximates the MMSE estimate $\hvec{x}\mmse$ and $\eta^{-1}$ approximates the resulting MMSE $\frac{1}{N}\tr(\Cov\{\vec{x}|\vec{y}\})$.

Although there is generally no closed-form expression for the moment-matching values of $(\vec{r}_1,\gamma_1,\vec{r}_2,\gamma_2,\hvec{x},\eta)$, one can iteratively solve for them using the EC algorithm shown in \algref{ec} (a form of expectation propagation (EP) \cite{Minka:Diss:01}) using the estimation functions
\begin{align}
\vec{f}_1(\vec{r}_1;\gamma_1)
&= \E\{\vec{x}|q_1\} 
= \textstyle \int \vec{x}\, q_1(\vec{x};\vec{r}_1;\gamma_1) \dif\vec{x}
\label{eq:f1_ec}\\
\vec{f}_2(\vec{r}_2;\gamma_2)
&= \E\{\vec{x}|q_2\} 
= \textstyle \int \vec{x}\, q_2(\vec{x};\vec{r}_2;\gamma_2) \dif\vec{x}
\label{eq:f2_ec}.
\end{align}
It is straightforward to show (see, e.g., \cite{Fletcher:ISIT:16}) that, at a fixed point of \algref{ec}, one obtains $\hvec{x}_1=\hvec{x}_2=\hvec{x}$ and $\eta_1=\eta_2=\eta=\gamma_1+\gamma_2$.

\begin{algorithm}[t]
\caption{EC / VAMP}
\label{alg:ec}
\begin{algorithmic}[1]
\Require{$\vec{f}_1(\cdot;\cdot)\text{~and~}\vec{f}_2(\cdot;\cdot)$.}
\State{Select initial $\vec{r}_1\in\Real^N,\gamma_1>0$}
\Repeat
    \State{// Measurement fidelity}
    \State{$\hvec{x}_1 \gets \vec{f}_1(\vec{r}_1;\gamma_1)$} \label{line:ec_x1}
    \State{$\eta_1 \gets \gamma_1 N / 
        \tr( \nabla\vec{f}_1(\vec{r}_1;\gamma_1) )$} \label{line:ec_eta1}
    \State{$\gamma_2 \gets \eta_1 - \gamma_1$}  \label{line:ec_gam2}
    \State{$\vec{r}_2 \gets (\eta_1\hvec{x}_1 - \gamma_1\vec{r}_1)/\gamma_2$}
        \label{line:ec_r2}
    \State{// Denoising}
    \State{$\hvec{x}_2 \gets \vec{f}_2(\vec{r}_2;\gamma_2)$} \label{line:ec_x2}
    \State{$\eta_2 \gets \gamma_2 N / 
        \tr( \nabla\vec{f}_2(\vec{r}_2;\gamma_2) )$} \label{line:ec_eta2}
    \State{$\gamma_1 \gets \eta_2 - \gamma_2$}  \label{line:ec_gam1}
    \State{$\vec{r}_1 \gets (\eta_2\hvec{x}_2 - \gamma_2\vec{r}_2)/\gamma_1$}
        \label{line:ec_r1}
\Until{Terminated}
\State \Return $\hvec{x}_2$
\end{algorithmic}
\end{algorithm}

For WGN-corrupted linear measurements $\vec{y}$ as in \eqref{y}, the likelihood becomes $\like(\vec{x};\vec{y}) = \mc{N}(\vec{y};\vec{Ax},\vec{I}/\gamma_w)$ and so $\vec{f}_1$ in \eqref{f1_ec} manifests as
\begin{align}
\vec{f}_1(\vec{r}_1;\gamma_1)
= \big(\gamma_w\vec{A}\herm\vec{A}+\gamma_1\vec{I}\big)^{-1}\big(\gamma_w\vec{A}\herm\vec{y} + \gamma_1\vec{r}_1\big)
\label{eq:f1_ec_awgn} .
\end{align}
This $\vec{f}_1$ can be interpreted as the MMSE estimator of $\vec{x}\true$ from the measurements $\vec{y}=\vec{Ax}\true+\mc{N}(\vec{0},\vec{I}/\gamma_w)$ under the pseudo-prior $\vec{x}\true\sim\mc{N}(\vec{r}_1,\vec{I}/\gamma_1)$.
Meanwhile $\vec{f}_2$ in \eqref{f2_ec} can be interpreted as the MMSE estimator of $\vec{x}\true$ from the pseudo-measurements $\vec{r}_2=\vec{x}\true+\mc{N}(\vec{0},\vec{I}/\gamma_2)$ under the prior $\vec{x}\true\sim\prior(\vec{x})$.
In other words, $\vec{f}_2$ can be interpreted as the MMSE denoiser of $\vec{r}_2$.
This pseudo-measurement model is exactly the same one that arises in AMP (recall \eqref{e_amp}).

For generic $\vec{A}$, there are no guarantees on the quality of the EC estimate $\hvec{x}$ or even the convergence of \algref{ec}.
But when $\vec{A}$ is a right orthogonally invariant (ROI) random matrix, EC has a rigorous high-dimensional analysis.
ROI matrices can be understood as those with singular value decompositions of the form $\vec{USV}\tran$, for orthogonal $\vec{U}$, diagonal $\vec{S}$, and random $\vec{V}$ uniformly distributed over the set of orthogonal matrices; the ROI class includes the i.i.d.\ Gaussian class but is more general.
In particular, \cite{Rangan:ISIT:17,Rangan:TIT:19} showed that, for asymptotically large ROI matrices $\vec{A}$, EC's denoiser input error $\vec{e}_2=\vec{r}_2-\vec{x}\true$ obeys
\begin{align}
\vec{e}_2 \sim \mc{N}(\vec{0},\vec{I}/\gamma_2)
\end{align}
at every iteration, similar to AMP (recall \eqref{e_amp}).
Likewise, macroscopic statistical quantities like MSE $\mc{E}=\frac{1}{N}\|\hvec{x}-\vec{x}\true\|^2$ obey a scalar state evolution. 
Importantly, these results hold not only for the MMSE denoising functions $\vec{f}_2$ specified by EC, but also for general Lipschitz $\vec{f}_2$ \cite{Rangan:TIT:19,Fletcher:NIPS:18}.
Due to the tight connections with AMP, the EC algorithm with general Lipschitz $\vec{f}_2$ was referred to as Vector AMP (VAMP) in \cite{Rangan:TIT:19,Fletcher:NIPS:18}.
A similar rigorous analysis of EC with asymptotically large, right unitarily invariant (RUI) matrices $\vec{A}$ was given in \cite{Takeuchi:ISIT:17}.
For those matrices, the SVD of $\vec{A}$ takes the form $\vec{USV}\herm$ with random $\vec{V}$ uniformly distributed over the set of unitary matrices.

Given that the EC/VAMP algorithm can be used with estimation functions other than the MMSE choices in \eqref{f1_ec}-\eqref{f2_ec}, one might wonder whether it can be applied to solve optimization problems of the form \eqref{opt}, i.e., MAP estimation.
This was answered affirmatively in \cite{Fletcher:ISIT:16}.
In particular, it suffices to choose 
\begin{align}
\vec{f}_1(\vec{r}_1,\gamma_1) &= \prox_{\gamma_1^{-1} g_1}(\vec{r}_1) \\
\vec{f}_2(\vec{r}_2,\gamma_2) &= \prox_{\gamma_2^{-1} g_2}(\vec{r}_2) .
\end{align}
Furthermore, the resulting EC/VAMP algorithm can be recognized as a form of ADMM.
If we fix the values of $\gamma_1$ and $\gamma_2$ over the iterations (which forces $\eta_1=\eta_2=\gamma_1+\gamma_2$) and define 
$\vec{u}_1\defn\gamma_1(\hvec{x}_2-\vec{r}_1)$ and
$\vec{u}_2\defn\gamma_2(\vec{r}_2-\hvec{x}_1)$,
we can rewrite EC/VAMP from \algref{ec} as the recursion 
\begin{subequations}\label{eq:admmprs2}
\begin{align}
\hvec{x}_1
&\leftarrow \prox_{\gamma_1^{-1} g_1}(\hvec{x}_2-\vec{u}_1/\gamma_1) \\
\vec{u}_2
&\leftarrow \vec{u}_1 + \gamma_1(\hvec{x}_1-\hvec{x}_2) \\
\hvec{x}_2
&\leftarrow \prox_{\gamma_2^{-1} g_2}(\hvec{x}_1+\vec{u}_2/\gamma_2) 
\label{eq:admmprs2_denoise}\\
\vec{u}_1
&\leftarrow \vec{u}_2 + \gamma_2(\hvec{x}_1-\hvec{x}_2)
\end{align}
\end{subequations}
which is a generalization of ADMM in \eqref{admm} to two dual updates and two penalty parameters.
If we additionally constrain $\gamma_1=\gamma_2\defn\gamma$ then \eqref{admmprs2} reduces to
\begin{subequations}\label{eq:admm_prs}
\begin{align}
\hvec{x}_1
&\leftarrow \prox_{\gamma^{-1} g_1}(\hvec{x}_2-\vec{u}) \\
\vec{u}
&\leftarrow \vec{u} + (\hvec{x}_1-\hvec{x}_2) \\
\hvec{x}_2
&\leftarrow \prox_{\gamma^{-1} g_2}(\hvec{x}_1+\vec{u}) 
\label{eq:admmprs_denoise}\\
\vec{u}
&\leftarrow \vec{u} + (\hvec{x}_1-\hvec{x}_2) ,
\end{align}
\end{subequations}
which is known as the Peaceman-Rachford or symmetric variant of ADMM, and which is said to converge faster than standard ADMM \cite{He:JO:14,He:JIS:16}.
The important point is that EC/VAMP can be understood as a generalization of ADMM that i) uses two penalty parameters and ii) adapts those penalty parameters with the iterations.

Inspired by D-AMP \cite{Metzler:ICIP:15}, a ``Denoising VAMP'' (D-VAMP) was proposed in \cite{Schniter:BASP:17}, which used VAMP with high-performance image denoisers and the Monte-Carlo approximation \eqref{trJfapprox}.
Although D-VAMP was shown to work well with large ROI $\vec{A}$, it can diverge with non-random $\vec{A}$, such as those encountered in MRI.
Some intuition behind the failure of VAMP with non-ROI $\vec{A}$ will be given in \secref{dgec}

\subsection{AMP/VAMP for MRI} \label{sec:mriamp}

The versions of $\vec{A}$ that manifest in linear inverse problems often do not have sufficient randomness for the AMP and EC/VAMP algorithms to work as intended.
If used without modification, AMP and EC/VAMP algorithms may simply diverge.
This is definitely the case for MRI, where $\vec{A}$ is the Fourier-based matrix shown in \eqref{A}.
Consequently, modified AMP and VAMP algorithms have been proposed specifically for MRI image recovery.

For example, \cite{Eksioglu:JIS:18} proposed to use D-AMP \eqref{amp} with $\beta\ll \sqrt{N}/\|\vec{A}\|_F$, which helps to slow down the algorithm and help it converge, but at the cost of degrading its fixed points, as we show in \secref{single}.
The authors of \cite{Sarkar:ICASSP:21} instead used damping to help D-VAMP converge without disturbing its fixed points. 
In conjunction with a novel initialization based on \textb{Peaceman-Rachford ADMM}, the latter scheme was competitive with PnP-ADMM for single-coil MRI.

\textb{For the special case of 2D point-sampled MRI,} the principle of density compensation \cite{Pipe:MRM:99} has also been exploited for the design of AMP-based algorithms.
For applications where k-space is non-uniformly sampled, density compensation applies a gain to each k-space sample that is proportional to the inverse sampling density at that sample, changing $\vec{y}$ to $\vec{G y}$ in \eqref{y} with diagonal gain matrix $\vec{G}$.
\textb{When $\vec{A}$ uses a 2D point mask,} the error in the density-compensated linear estimate $\hvec{x}=\vec{A}\herm\vec{G y}$ behaves much more like white Gaussian noise than does the error in the standard linear estimate $\hvec{x}=\vec{A}\herm\vec{y}$ (see, e.g., \cite{Edupuganti:TMI:20}).
After observing the error to behave even more like white noise within wavelet subbands, Millard et al.\ \cite{Millard:OJSP:20} proposed a VAMP modification that employs density compensation in the linear stage and wavelet thresholding in the denoising stage.
The resulting ``Variable-Density AMP'' (VDAMP) algorithm was empirically observed to successfully track the error variance in each subband over the algorithm iterations.
The authors then extended their work from single- to multicoil MRI in \cite{Millard:22}, calling their approach Parallel VDAMP (P-VDAMP). 

To improve on VDAMP, Metzler and Wetzstein \cite{Metzler:ICASSP:21} proposed a PnP extension of the algorithm, where the wavelet-thresholding denoiser was replaced by a novel DNN that accepts a vector of subband error variances at each iteration.
The resulting Denoising VDAMP (D-VDAMP) showed a significant boost in recovery accuracy over VDAMP for single-coil \textb{2D point-sampled} MRI \cite{Metzler:ICASSP:21}.
Although D-VDAMP works relatively well, it requires early stopping for good performance (as we demonstrate in \secref{single}), which suggests that D-VDAMP has suboptimal fixed points and hence can be improved. 
\textb{Most recently, a ``Denoising P-VDAMP'' (DP-VDAMP) was proposed \cite{Millard:ISMRM:22,Millard:Diss:21} that replaces the wavelet thresholding step in P-VDAMP with a DNN denoiser.
A major shortcoming of VDAMP, P-VDAMP, D-VDAMP, and DP-VDAMP is that they are designed around the use of 2D point sampling masks, which are impractical and uncommon in clinical MRI.
These shortcomings} motivate our proposed approach, which is described in the next section.


\section{Proposed Approach}

We now propose a new approach to MRI recovery that, like \textb{the VDAMP-based algorithms \cite{Millard:OJSP:20,Millard:22,Metzler:ICASSP:21,Millard:ISMRM:22}}, formulates signal recovery in the wavelet domain, but, unlike the VDAMP-based algorithms, does not use density compensation \textb{and does not require the use of 2D point masks}. 
Our approach is based on a PnP version of the generalized EC algorithm, which is described in \secref{dgec}, in conjunction with a DNN denoiser that can handle parameterized colored noise, which is described in \secref{corr+corr}. 

\subsection{Wavelet-domain denoising GEC algorithm} \label{sec:dgec}

To motivate wavelet-domain signal recovery, we first present an intuitive explanation of the problems faced by EC/VAMP with non-ROI $\vec{A}$.
To start, one can show (see \iftoggle{include_app}{\appref{recursion}}{Appendix A}) that EC/VAMP's denoiser input error $\vec{e}_2\defn\vec{r}_2-\vec{x}\true$ can be written as
\begin{align}
\ebf_2
&= \Vbf\vec{D}\Vbf\herm \ebf_1 + \ubf 
\label{eq:ebf2} ,
\end{align}
where
$\vec{V}$ is the right singular vector matrix of $\vec{A}$,
the matrix $\vec{D}$ is diagonal with $\tr(\vec{D})=0$,
$\vec{e}_1\defn\vec{r}_1-\vec{x}\true$ is the error on the input to $\vec{f}_1$,
and
$\vec{u}$ is a linear transformation of the measurement noise vector $\vec{w}$ from \eqref{y}.
When $\vec{A}$ is ROI or RUI, $\vec{V}$ is drawn uniformly from the group of orthogonal or unitary matrices, respectively. 
\iftoggle{include_app}{\appref{ebf2}}{Appendix B} shows for the orthogonal case that, 
if $\Vbf$ and $\ebf_1$ are treated as independent up to the fourth moment
and $\wbf$ and $\ebf_1$ are uncorrelated,
then, conditioned on $\ebf_1$, both $\Vbf\vec{D}\Vbf\herm \ebf_1$ and
$\ebf_2$ are asymptotically white and \text{zero-mean} Gaussian.
Importantly, this behavior occurs despite the tendency for $\ebf_1$ to be highly structured and non-Gaussian.

When $\vec{A}$ is not a high-dimensional ROI or RUI matrix, however, there is no guarantee that 
\textb{$\Vbf\vec{D}\Vbf\herm\ebf_1$ will asymptotically be white and zero-mean Gaussian.}
For example, when $\vec{A}=\vec{MF}$ as in single-coil MRI and $\vec{x}\true$ is a natural image, this 
desired 
\textb{property does not manifest}
because the $\vec{x}\true$ (and thus $\ebf_1$) has a high concentration of energy at low frequencies and $\vec{V}\herm=\vec{F}$ focuses that error into a few dimensions of $\vec{D}$.

\begin{figure}[t]
 \centering
  \includegraphics[width = \columnwidth,trim={110 0 100 0},clip]{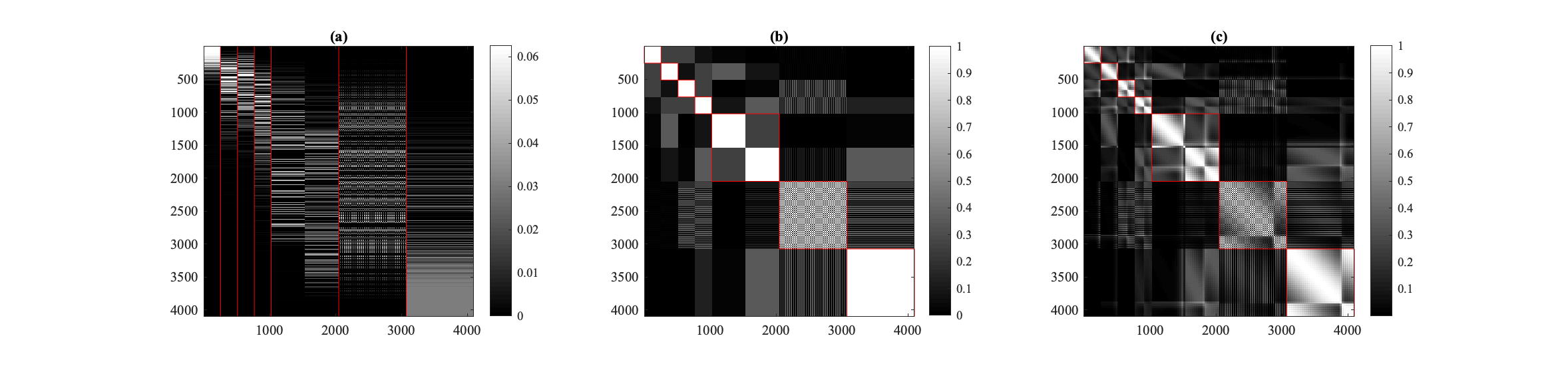}
  \caption{\blue Approximate block-diagonality of 2D Fourier-wavelet matrices. 
  Using $\abs(\cdot)$ to denote the entry-wise magnitude operation, (a) shows $\abs(\vec{F\Psi}\tran)$ with rows sorted according to distance from the k-space origin, columns sorted according to wavelet subbands, and subband boundaries denoted by red lines. 
  Meanwhile, (b) shows the matrix product $\abs(\vec{F\Psi}\tran)\tran\abs(\vec{F\Psi}\tran)$ and (c) shows $\abs(\vec{G})\tran\abs(\vec{G})$ for the multi-coil Fourier-wavelet matrix $\vec{G}$ defined in \eqref{G}.
  The approximate block-diagonality of (b) and (c) suggests that the columns of the 2D Fourier-wavelet matrices are well decoupled in the single- and multi-coil cases.}
  \label{fig:fourier_wavelet}
\end{figure}

We now explain why using an AMP/EC algorithm to recover the wavelet coefficients $\vec{c}\true \defn \vec{\Psi}\vec{x}\true$, rather than the image pixels $\vec{x}\true$, offers a path to circumvent these issues.
For an orthogonal discrete wavelet transform (DWT) $\vec{\Psi}$, we have $\vec{x}\true=\vec{\Psi}\tran\vec{c}\true$ and so \eqref{y} implies the measurement model
\begin{align}
\vec{y} = \vec{B}\vec{c}\true + \vec{w}
\text{~~with~~}
\vec{B} \defn \vec{A\Psi}\tran 
\label{eq:y2} .
\end{align}
In the case where $\vec{A}$ is a subsampled version of the Fourier matrix $\vec{F}$, the matrix $\vec{B}$ is a subsampled Fourier-wavelet matrix $\vec{F\Psi}\tran$.
The Fourier-wavelet matrix is known to be approximately block diagonal \textb{after appropriate row-sorting} \cite{Adcock:FMS:17}, where the blocks correspond to the wavelet subbands.
\textb{This means that $\vec{B}$ in \eqref{y2} primarily mixes the wavelet coefficients $\vec{c}\true$ \emph{within} subbands rather than \emph{across} subbands.
Consequently, if that mixing has a sufficiently randomizing effect on each subband of $\ebf_1$}, then---with an appropriate EC-style algorithm design---the subband error vectors $\ebf_2$ can be kept approximately i.i.d.\ Gaussian across the iterations, although with a possibly different variance in each subband. 
\textb{In \figref{fourier_wavelet}(a), we plot $\abs(\vec{F\Psi}\tran)$ for the 2D case with the rows sorted according to the distance of their corresponding k-space sample to the origin. 
Although this row-sorting does not yield an approximately block-diagonal matrix, it should be clear from the discussion above that row-sorting is unimportant;
it only matters that the columns of $\vec{B}$ for each given subband have a sufficiently randomizing effect on that subband and are approximately decoupled from the columns of other subbands.
To illustrate the degree of column-decoupling in $\vec{F\Psi}\tran$, we plot $\abs(\vec{F\Psi}\tran)\tran\abs(\vec{F\Psi}\tran)$ in \figref{fourier_wavelet}(b).
We plot this particular quantity because, if $\vec{F\Psi}\tran=\vec{JD}$ where $\vec{J}$ is a permutation matrix and $\vec{D}$ is a perfectly block-diagonal matrix, then $\abs(\vec{F\Psi}\tran)\tran\abs(\vec{F\Psi}\tran)$ will be perfectly block-diagonal for \emph{any} $\vec{J}$, i.e., for any row-sorting.
The fact that \figref{fourier_wavelet}(b) looks approximately block-diagonal suggests that the column-blocks of $\vec{F\Psi}\tran$ are significantly decoupled.} 

\textb{The discussion in the previous paragraph pertains to single-coil MRI.
In the multi-coil case, the matrix $\vec{A}$ takes the form in \eqref{A} and so $\vec{B}$ from \eqref{y2} manifests as
\begin{align}
\vec{B} 
= \mat{\vec{M} & &\\[-2mm]
               & \ddots &\\[-2mm]
               & & \vec{M}}
  \vec{G}
  \text{~~with~~}
  \vec{G}\defn
  \mat{\vec{F}\Diag(\vec{s}_1)\vec{\Psi}\tran\\[-2mm]\vdots\\[-2mm]\vec{F}\Diag(\vec{s}_C)\vec{\Psi}\tran}
\label{eq:G} .
\end{align}
We would like that the multi-coil Fourier-wavelet matrix $\vec{G}$ has a sufficiently randomizing effect on each given subband in $\vec{c}_0$ and that the columns corresponding to that subband are decoupled from the columns of other subbands.
To investigate the decoupling behavior of $\vec{G}$, we plot $\abs(\vec{G})\tran\abs(\vec{G})$ in \figref{fourier_wavelet}(c) for the case of $C=8$ ESPIRiT-estimated coils and notice that, similar to the single-coil quantity $\abs(\vec{F\Psi}\tran)\tran\abs(\vec{F\Psi}\tran)$ in \figref{fourier_wavelet}(b), the multi-coil quantity $\abs(\vec{G})\tran\abs(\vec{G})$ looks approximately block-diagonal.
}

The first 
\textb{AMP-based method that exploited the aforementioned Fourier-wavelet properties} 
was the VAMPire algorithm from \cite{Schniter:BASP:17b}, where a normalization of the subband energies in $\vec{c}\true$ was used to equalize the subband error variances in $\vec{e}_2$, with the goal of tracking a single variance across the iterations (thus facilitating the use of D-VAMP). 
In other words, \eqref{y2} was written as $\vec{y}=\ovec{B}\ovec{c}\true+\vec{w}$ with $\ovec{B}=\vec{B}\Diag(\vec{g})$ and $\ovec{c}\true=\Diag(\vec{g})^{-1}\vec{c}\true$, for $\vec{g}$ such that $\diag(\Cov(\ovec{c}\true))\approx\vec{1}$.
But, because the variances of the subbands in $\vec{e}_2$ do change with the iterations, the scheme in \cite{Schniter:BASP:17b} was far from optimal.

\begin{algorithm}[t]
\caption{Generalized EC (GEC)}
\label{alg:gec}
\begin{algorithmic}[1]
\Require{$\vec{f}_1(\cdot;\cdot),~\vec{f}_2(\cdot;\cdot),\text{ and }\gdiag(\cdot)$.}
\State{Select initial $\vec{r}_1,\vec{\gamma}_1$}
\Repeat
    \State{// Measurement fidelity}
    \State{$\hvec{x}_1 \gets \vec{f}_1(\vec{r}_1,\vec{\gamma}_1)$} \label{line:gec_x1}
    \State{$\vec{\eta}_1 \gets 
        \Diag( \gdiag( \nabla\vec{f}_1(\vec{r}_1,\vec{\gamma}_1) ))^{-1} 
        \vec{\gamma}_1$} \label{line:gec_eta1}

    \State{$\vec{\gamma}_2 \gets \vec{\eta}_1 - \vec{\gamma}_1$}  \label{line:gec_gam2}
    \State{$\vec{r}_2 \gets \Diag(\vec{\gamma}_2)^{-1}(\Diag(\vec{\eta}_1)\hvec{x}_1 - \Diag(\vec{\gamma}_1)\vec{r}_1)$}
        \label{line:gec_r2}
    \State{// Denoising}
    \State{$\hvec{x}_2 \gets \vec{f}_2(\vec{r}_2,\vec{\gamma}_2)$} \label{line:gec_x2}
    \State{$\vec{\eta}_2 \gets 
        \Diag( \gdiag( \nabla\vec{f}_2(\vec{r}_2,\vec{\gamma}_2) ))^{-1} 
        \vec{\gamma}_2$} \label{line:gec_eta2}
    \State{$\vec{\gamma}_1 \gets \vec{\eta}_2 - \vec{\gamma}_2$}  \label{line:gec_gam1}
    \State{$\vec{r}_1 \gets \Diag(\vec{\gamma}_1)^{-1}(\Diag(\vec{\eta}_2)\hvec{x}_2 - \Diag(\vec{\gamma}_2)\vec{r}_2)$}
        \label{line:gec_r1}
\Until{Terminated}
\State \Return $\hvec{x}_2$
\end{algorithmic}
\end{algorithm}

In this work, we propose an EC-based PnP method that recovers the wavelet coefficients $\vec{c}\true$ and tracks the variances of both $\vec{e}_1$ and $\vec{e}_2$ in each wavelet subband. 
Our approach leverages the Generalized EC (GEC) framework from \cite{Fletcher:ISIT:16}, which is summarized in \algref{gec} and \eqref{gdiag}.
GEC is a generalization of EC from \algref{ec} that averages the diagonal of the Jacobian $\nabla\vec{f}_i$ separately over $L$ coefficient subsets using the $\gdiag\!:\Real^{N\times N}\!\rightarrow\!\Real^N$ operator:
\begin{subequations}
\label{eq:gdiag} 
\begin{eqnarray}
\gdiag(\vec{Q})
&\defn& [d_1\vec{1}_{N_1}\tran,\dots,d_L\vec{1}_{N_L}\tran]\tran \qquad 
\label{eq:gdiag_vec}\\
d_\ell 
&=& \frac{\tr\{\vec{Q}_{\ell\ell}\}}{N_\ell} 
\label{eq:gdiag_tr}.
\end{eqnarray}
\end{subequations}
In \eqref{gdiag}, $N_\ell$ denotes the size of the $\ell$th subset
and $\vec{Q}_{\ell\ell}\in\Real^{N_\ell\times N_\ell}$ denotes the $\ell$th diagonal subblock of the matrix input $\vec{Q}$.
%
When GEC is used to solve a convex optimization problem of the form \eqref{opt}, the functions $\vec{f}_i$ take the form
\begin{align}
\vec{f}_i(\vec{r},\vec{\gamma}) 
&= \gprox_{g_i,\vec{\gamma}}(\vec{r})
\quad\text{for}\quad
\gprox_{\rho,\vec{\gamma}}(\vec{r}) 
\defn \arg\min_{\vec{x}} \big\{ \rho(\vec{x}) + \tfrac{1}{2}\|\vec{x}-\vec{r}\|^2_{\vec{\gamma}}\big\}
\label{eq:gprox} ,
\end{align}
where $\|\vec{q}\|_{\vec{\gamma}} \defn \sqrt{\vec{q}\herm\Diag(\vec{\gamma})\vec{q}}$.
When $L\!=\!1$, GEC reduces to EC/VAMP.
In that case, $\vec{\gamma}=\gamma\vec{1}$ and $\gprox_{\rho,\vec{\gamma}}=\prox_{\gamma^{-1}\rho}$.

\begin{algorithm}[t]
\caption{Denoising GEC operating in the wavelet domain}
\label{alg:dgec}
\begin{algorithmic}[1]
\Require{$\vec{f}_1(\cdot,\cdot),~\vec{f}_2(\cdot,\cdot),~\gdiag(\cdot),\text{ and }\vec{\Psi}$.}
\State{Select initial $\vec{r}_1,\vec{\gamma}_1$}
\Repeat
    \State{// Measurement fidelity}
    \State{$\hvec{c}_1 \gets \vec{f}_1(\vec{r}_1,\vec{\gamma}_1)$} \label{line:dgec_x1}
    \State{$\vec{\eta}_1 \gets 
        \Diag( \gdiag( \nabla\vec{f}_1(\vec{r}_1,\vec{\gamma}_1) ))^{-1} 
        \vec{\gamma}_1$} \label{line:dgec_eta1}
    \State{$\vec{\gamma}_2 \gets \vec{\eta}_1 - \vec{\gamma}_1$}  \label{line:dgec_gam2}
    \State{$\vec{r}_2 \gets \Diag(\vec{\gamma}_2)^{-1}(\Diag(\vec{\eta}_1)\hvec{c}_1 - \Diag(\vec{\gamma}_1)\vec{r}_1)$}
        \label{line:dgec_r2}
    \State{// Denoising}
    \State{$\hvec{c}_2 \gets \vec{\Psi}\vec{f}_2(\vec{\Psi}\tran \vec{r}_2,\vec{\gamma}_2)$} \label{line:dgec_x2}
    \State{$\vec{\eta}_2 \gets 
        \Diag( \gdiag( \nabla\vec{f}_2(\vec{r}_2,\vec{\gamma}_2) ))^{-1} 
        \vec{\gamma}_2$} \label{line:dgec_eta2}
    \State{$\vec{\gamma}_1 \gets \vec{\eta}_2 - \vec{\gamma}_2$}  \label{line:dgec_gam1}
    \State{$\vec{r}_1 \gets \Diag(\vec{\gamma}_1)^{-1}(\Diag(\vec{\eta}_2)\hvec{c}_2 - \Diag(\vec{\gamma}_2)\vec{r}_2)$}
        \label{line:dgec_r1}
\Until{Terminated}
\State \Return $\hvec{x}_2 = \vec{\Psi}\tran\hvec{c}_2$
\end{algorithmic}
\end{algorithm}

Our proposed wavelet-domain Denoising GEC (D-GEC) approach is outlined in \algref{dgec}.
For the $\gdiag$ operator, we use \eqref{gdiag} with the diagonalization subsets defined by the $L=3D+1$ subbands of a depth-$D$ dyadic 2D orthogonal DWT.
Also, when computing $\gdiag(\nabla\vec{f}_1)$ and $\gdiag(\nabla\vec{f}_2)$ 
in lines~\ref{line:dgec_eta1}~and~\ref{line:dgec_eta2}, we approximate the $\tr\{\vec{Q}_{\ell\ell}\}$ terms in \eqref{gdiag_tr} using the Monte Carlo approach \textb{\cite{Ramani:TIP:08}}
\begin{eqnarray}
\tr\{\vec{Q}_{\ell\ell}\}
&\approx& \delta_{\ell}^{-1}\vec{q}_\ell\herm\big[\vec{f}_i(\vec{r}+\delta_{\ell}\vec{q}_\ell,\vec{\gamma})-\vec{f}_i(\vec{r},\vec{\gamma})\big]
\label{eq:trJfapprox_ll} ,
\end{eqnarray}
where we use i.i.d.\ unit-variance Gaussian coefficients for the $\ell$th coefficient subset in $\vec{q}_\ell$ and set all other coefficients in $\vec{q}_\ell$ to zero.
As a result of the chosen diagonalization, the $\vec{\gamma}_i$ vectors (for $i=1,2$) are structured as
\begin{align}
\vec{\gamma}_i
= [\gamma_{i,1}\vec{1}_{N_1}\tran,\dots,\gamma_{i,L}\vec{1}_{N_L}\tran]\tran
\label{eq:gam} ,
\end{align}
and the $\vec{\eta}_i$ vectors have a similar structure.
\textb{In \eqref{trJfapprox_ll} we used $\delta_{\ell}=\min\{\sqrt{1/\gamma_{\ell}},\|\vec{r}_{\ell}\|_1/N_\ell\}$ where $\vec{r}_{\ell}$ denotes the $\ell$th coefficient subset of $\vec{r}$.}

For the wavelet-measurement model \eqref{y2} with WGN $\vec{w}$, \eqref{gprox} implies that the $\vec{f}_1$ estimation function in \lineref{dgec_x1} of \algref{dgec} manifests as
\begin{align}
\vec{f}_1(\vec{r}_1,\vec{\gamma}_1) 
&= \big(\gamma_w\vec{B}\herm\vec{B}+\Diag(\vec{\gamma}_1)\big)^{-1}\big(\gamma_w\vec{B}\herm\vec{y}+\Diag(\vec{\gamma}_1)\vec{r}_1\big)
\label{eq:f1} .
\end{align}
When numerically solving \eqref{f1}, we exploit the fact that $\vec{B}$ is a fast operator by using the conjugate gradient (CG) method \cite{Golub:Book:96}.

For $\vec{f}_2$ in \lineref{dgec_x2} of \algref{dgec}, we use a pixel-domain DNN denoiser.
As shown in \lineref{dgec_x2}, we convert from the wavelet domain to the pixel domain and back when calling this denoiser.
Note that the denoiser $\vec{f}_2$ is provided with the vector $\vec{\gamma}_2$ of subband error precisions.
The design of this denoiser will be discussed in \secref{corr+corr}.
The experiments in \secref{example} suggest that the denoiser input error $\vec{e}_2=\vec{r}_2-\vec{c}\true$ does indeed obey 
\begin{align}
\vec{e}_2 
\sim \mc{N}(\vec{0},\Diag(\vec{\gamma}_2)^{-1})
\label{eq:dgec_e2}
\end{align}
for the $\vec{\gamma}_2$ vector computed in \lineref{dgec_gam2} of \algref{dgec}, similar to other AMP, VAMP, EC, and GEC algorithms.
Further work is needed to understand if this behavior can be predicted by a rigorous analysis.
The error model \eqref{dgec_e2} facilitates a principled way to train the DNN denoiser, as we discuss in the next section. 

\blue
We now discuss the initialization of D-GEC.
For \eqref{dgec_e2} to hold at all iterations, we 
need that the initial $\vec{\gamma}_1$ contains the precisions (i.e., inverse variances) of the subbands of the initial $\vec{e}_1=\vec{r}_1-\vec{c}\true$.
But initializing $\vec{\gamma}_1$ is complicated by the fact that $\vec{c}\true$ is unknown.
In response, we suggest initializing $\vec{\gamma}_1$ at an \emph{average} value such as 
\begin{align}
\hvec{\gamma}_1
= \Diag(\gdiag(\E\{(\vec{r}_1-\vec{c}\true)(\vec{r}_1-\vec{c}\true)\herm\}))^{-1}\vec{1}
\label{eq:gam_init},
\end{align}
where the expectation is approximated using a sample average over a training set (e.g., the dataset used to train the denoiser).
But this approach could fail if the precision of the initial error falls far from $\hvec{\gamma}_1$, which can happen if $\vec{r}_1$ is strongly dependent on $\vec{y}$.
Thus, we propose to initialize $\vec{r}_1=\vec{B}\herm\vec{y} + \vec{n}$, where $\vec{n}$ is Gaussian and white in each subband.
The per-subband variance of $\vec{n}$ should be large enough to dominate the behavior of $\vec{e}_1$, which makes the subband precisions easy to predict, but not so large that the algorithm is initialized at a terribly bad state.
For the experiments in \secref{example}, we set the per-subband variance of $\vec{n}$ at $10$ times the per-subband variance of $\vec{B}\herm\vec{y}-\vec{c}\true$, and observed that \eqref{dgec_e2} held at all iterations.
Although a careful choice of initialization is important for \eqref{dgec_e2} to hold at all iterations, we find that the initialization has little effect on the fixed points of D-GEC.
So, for the experiments in Sections~\ref{sec:multi_point},~\ref{sec:multi_line}, and \ref{sec:single}, we set $\vec{n}=\vec{0}$ to improve the accuracy of the initial $\vec{r}_1$ and thus speed D-GEC convergence.
\color{black}

\textb{Computationally, the cost of D-GEC is driven by lines~\ref{line:dgec_x1}-\ref{line:dgec_eta1} and \ref{line:dgec_x2}-\ref{line:dgec_eta2} of \algref{dgec}, which call $\vec{f}_1$ and $\vec{f}_2$, respectively, $L+1$ times when implementing \eqref{trJfapprox_ll}.
The $L+1$ calls to $\vec{f}_1$ can be performed in parallel (e.g., in a single minibatch on a GPU), as can the calls to $\vec{f}_2$.
As described above, each call to $\vec{f}_1$ involves running several iterations of CG.
For accurate D-GEC fixed points, we find that $10$ CG iterations suffice, and we use this setting in Sections~\ref{sec:multi_point},~\ref{sec:multi_line}, and \ref{sec:single}.
For D-GEC error to match the state-evolution predictions at all iterations, we find that $150$ CG iterations suffice, and we use this value in \secref{example}.
Each call to $\vec{f}_2$ involves calling the DNN denoiser that is described in the next subsection.}

\subsection{A DNN denoiser for correlated noise} \label{sec:corr+corr}

As suggested by \eqref{dgec_e2}, the denoiser $\vec{f}_2$ in \algref{dgec} faces the task of denoising the pixel-domain signal $\vec{\Psi}\tran\vec{r}_2$, where $\vec{r}_2=\vec{c}\true+\textb{\vec{n}~\text{for}~\vec{n}\sim}\mc{N}(\vec{0},\Diag(\vec{\gamma}_2)^{-1})$ and $\vec{c}\true$ are the wavelet coefficients of the true image $\vec{x}\true$.
The denoiser input can thus be modeled as
\begin{align}
\vec{\Psi}\tran\vec{r}_2 
= \vec{x}\true + \textb{\vec{n}~\text{for}~\vec{n}\sim}\mc{N}(\vec{0},\vec{\Psi}\tran\Diag(\vec{\gamma}_2)^{-1}\vec{\Psi})
\label{eq:corr},
\end{align}
i.e., the true image corrupted by colored Gaussian noise with (known) covariance matrix $\vec{\Psi}\tran\Diag(\vec{\gamma}_2)^{-1}\vec{\Psi}$.
Here, the $\vec{\gamma}_2$ vector takes the form shown in \eqref{gam}.

Although several DNNs have been proposed to tackle denoising with correlated noise (e.g., \cite{Ahmadzadegan:SRep:21,Chang:TIM:19,Tiirola:JVCIR:19}), to our knowledge, the only one compatible with our denoising task is the DNN proposed by Metzler and Wetzstein in \cite{Metzler:ICASSP:21}.
There, they built on the DnCNN network by providing every layer with $L$ additional channels, where the $\ell$th channel contains the standard deviation (SD) of the noise in the $\ell$th wavelet subband (i.e., $\sqrt{1/\gamma_{2,\ell}}$).
Their approach can be interpreted as an extension of FFDNet \cite{Zhang:TIP:08}, which provides one additional channel containing the SD of the assumed white corrupting noise, to multiple additional channels containing subband SDs.
In our numerical experiments in \secref{numerical}, we find that Metzler's denoising approach works well in some cases but poorly in others.
We believe that the observed poor performance may be the result of the fact that their DNN operates in the pixel domain, while their SD side information is given in the wavelet domain and the network is given no information about the wavelet transform $\vec{\Psi}$.

We now propose a novel approach to DNN denoising that can handle colored Gaussian noise with an arbitrary known covariance matrix.
Our approach starts with an arbitrary DNN denoiser (e.g., DnCNN \cite{Zhang:TIP:17}, UNet \cite{Ronneberger:MICCAI:15}, RNN \cite{Zhang:ICLR:19}, etc.) that normally accepts $C$ input channels (e.g., 3 channels for color-image denoising or 2 channels for complex-image denoising).
It then adds $K\geq 1$ sets of $C$ additional channels, where each set is fed an independently generated realization of noise with the same statistics as that corrupting the signal to be denoised.
In other words, if $\vec{u}\in\Real^{CN}$ denotes the (vectorized) noisy input signal, which obeys (recall \eqref{corr})
\begin{align}
\vec{u} = \vec{x}\true + \textb{\vec{n}~\text{for}~\vec{n}\sim}\mc{N}(\vec{0},\vec{\Sigma}) 
\end{align}
with arbitrary known $\vec{\Sigma}$, then the (vectorized) input to the $k$th additional channel-set would be 
\begin{align}
\vec{n}_k \sim \mc{N}(\vec{0},\vec{\Sigma}) ~~\forall k=1,\dots,K,
\end{align}
where $\{\vec{n}_k\}_{k=1}^K$ are mutually independent and independent of $\vec{u}$.
The hope is that, during training, the denoiser learns how to i) extract the relevant statistics from $\{\vec{n}_k\}_{k=1}^K$ and ii) use them productively for the denoising of $\vec{u}$.
Here, $K$ is a design parameter; for our D-GEC application we find that $K=1$ suffices.
Because the denoiser accepts a signal corrupted by correlated noise plus additional realizations of correlated noise, we call our approach ``corr+corr.''

To train our corr+corr denoiser, we use the following approach.
Suppose that we have access to a training set of clean signals $\{\vec{x}_i\}$, and that we would like to train the denoiser to handle $\vec{\gamma}_2$ vectors from some distribution $p_\Gamma$.
During training, we draw many $\vec{\gamma}_2\sim p_\Gamma$ and, for each realization of $\vec{\gamma}_2$, we draw independent realizations of $\vec{v}$ and $\{\vec{n}_k\}_{k=1}^N$ from the distribution $\mc{N}(\vec{0},\vec{\Psi}\tran\Diag(\vec{\gamma}_2)^{-1}\vec{\Psi})$.
The $\vec{v}$ vector is then used to form the noisy signal $\vec{u}_i=\vec{x}_i+\vec{v}$ and the denoiser is given access to $\vec{N}\defn[\vec{n}_1,\dots,\vec{n}_K]$ when denoising $\vec{u}_i$.
Concretely, if we denote the corr+corr denoiser as $\vec{f}_2(\vec{u}_i,\vec{N};\vec{\theta})$, where $\vec{\theta}$ contains the trainable denoiser parameters, then we train those parameters using
\begin{align}
\hvec{\theta}
= \arg\min_{\vec{\theta}} \sum_i E\big\{  \mc{L}\big(\vec{x}_i,\vec{f}_2(\vec{x}_i+\vec{v},\vec{N};\vec{\theta})\big) \big\}
\label{eq:training},
\end{align}
where $\mc{L}(\cdot,\cdot)$ is a loss function that quantifies the error between its two vector-valued arguments.
Popular losses include \cite{Zhao:TCI:16} $\vec{\ell}_2$, $\vec{\ell}_1$, SSIM \cite{Wang:TIP:04}, or combinations thereof,
\textb{and in our experiments we used $\ell_2$ loss.}
The expectation in \eqref{training} is taken over both $\vec{v}$ and $\vec{N}$, which implicitly involves $p_\Gamma$.

In inference mode, we are given a noisy $\vec{u}$ and a single precision vector $\vec{\gamma}_2$.
From the latter, we generate a single independent realization of $\vec{N}\sim\mc{N}(\vec{0},\vec{\Psi}\tran\Diag(\vec{\gamma}_2)^{-1}\vec{\Psi})$
and then compute the denoised pixel-domain image estimate via
$\hvec{x}_2 = \vec{f}_2(\vec{u},\vec{N};\hvec{\theta})$.

In \secref{denoising} we show that our corr+corr denoiser performs better than Metzler's DnCNN and nearly as well as a genie-aided denoiser that knows the distribution of the test noise $\vec{v}\sim \vec{\Psi}\tran\Diag(\vec{\gamma}_2)^{-1}\vec{\Psi}$, with fixed $\vec{\gamma}_2$, at training time.


\section{Numerical Experiments} \label{sec:numerical}

In this section, we present numerical experiments demonstrating the performance of the proposed corr+corr denoiser as well as the proposed D-GEC method applied to both single-coil and multicoil MRI recovery.

\subsection{Denoising experiments} \label{sec:denoising}

In this subsection, we compare the corr+corr denoiser proposed in \secref{corr+corr} to several existing denoisers.
We test all denoisers on the 10 MRI images from the Stanford 2D FSE dataset \cite{Ong:ISMRM:18} shown in \figref{test_images}, which ranged in size from $320\times 320$ to $416\times 416$.
Noisy images were obtained by corrupting those test images by additive zero-mean Gaussian noise of covariance
\begin{align}
\vec{\Sigma} = \vec{\Psi}\tran\Diag(\vec{\gamma})^{-1}\vec{\Psi} 
\label{eq:test},
\end{align}
with $\vec{\Psi}$ a 2D Haar wavelet transform of depth $D=1$.
This wavelet transform has $L=4$ subbands, and so the precision vector $\vec{\gamma}$ in \eqref{test} is structured as $\vec{\gamma}=[\gamma_1 \vec{1}\tran_{N/4},\dots,\gamma_4\vec{1}\tran_{N/4}]\tran$ and thus parameterized by the four precisions $[\gamma_1,\gamma_2,\gamma_3,\gamma_4]$, or equivalently the four SDs $\big[\tfrac{1}{\sqrt{\gamma_1}},\tfrac{1}{\sqrt{\gamma_2}},\tfrac{1}{\sqrt{\gamma_3}},\tfrac{1}{\sqrt{\gamma_4}}\big]$.
We test the denoisers under different assumptions on these SDs, as indicated by the rows in \tabref{results_denoiser}.
For some tests, we use a fixed SD vector, while for other tests we average over a distribution of SD vectors.

\begin{figure}[t]
\centering
\newcommand{\wid}{0.19\columnwidth}
\includegraphics[width=\wid]{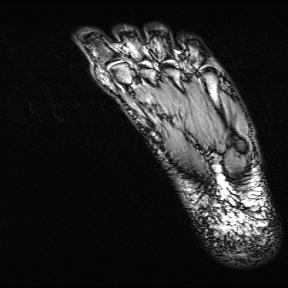}
\includegraphics[width=\wid]{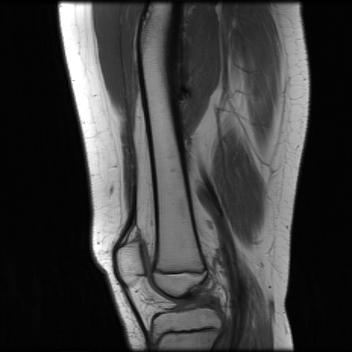}
\includegraphics[width=\wid]{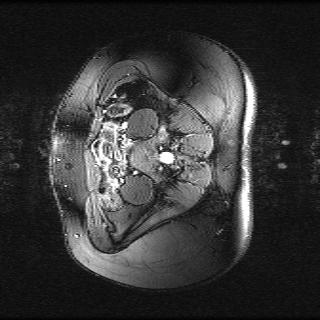}
\includegraphics[width=\wid]{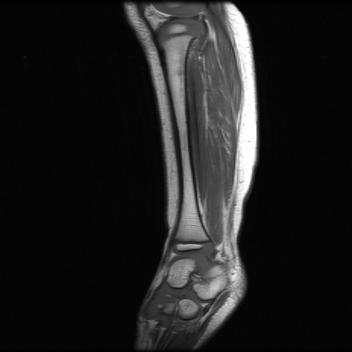}
\includegraphics[width=\wid]{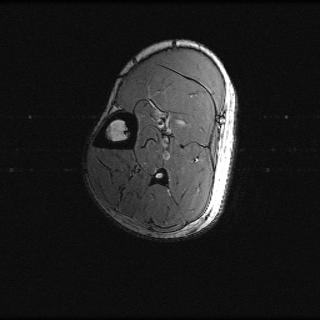}\mbox{}\\[1mm]
\includegraphics[width=\wid]{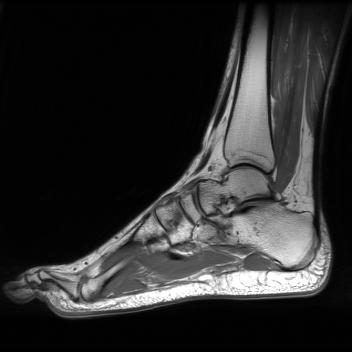}
\includegraphics[width=\wid]{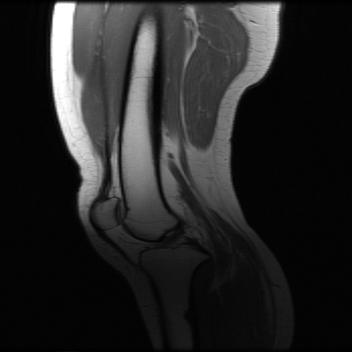}
\includegraphics[width=\wid]{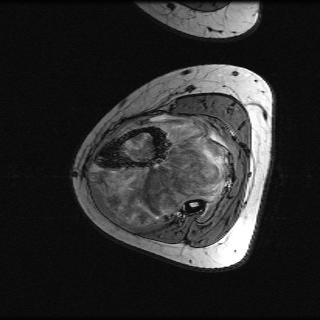}
\includegraphics[width=\wid]{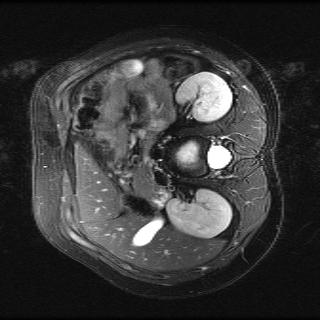}
\includegraphics[width=\wid]{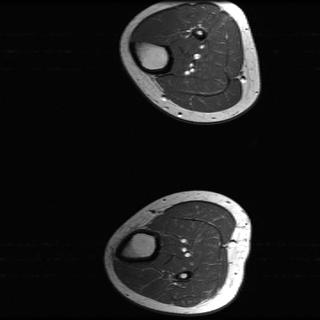}
\caption{Test images from the Stanford 2D FSE MRI dataset \cite{Ong:ISMRM:18}.}
\label{fig:test_images}
\end{figure}

When training the denoisers, we used the $70$ training MRI images from the Stanford 2D FSE dataset.
We trained to minimize \textb{$\ell_2$} loss on a total of $44\,000$ patches of size $40\times 40$ taken with stride $10\times 10$.
All denoisers used the bias-free version of DnCNN from \cite{Mohan:ICLR:20}, with the exception of Metzler's DnCNN from \cite{Metzler:ICASSP:21}, which used the publicly available code provided by the author.
For both corr+corr and Metzler's DnCNN, when training, we used random subband SDs $\{1/\sqrt{\gamma_\ell}\}_{\ell=1}^4$ drawn independently from a uniform distribution over the interval $[0,50/255]$.
When interpreting the value ``$50/255$,'' note that the image pixel values were in $[0,1]$ for this dataset.
As a baseline method, we trained bias-free DnCNN using white noise with a standard deviation distributed uniformly over the interval $[0,50/255]$.
We expect this ``white DnCNN'' to perform poorly with colored testing noise.
As an upper bound on performance, we trained bias-free DnCNN using the same fixed value of the SD vector $\big[\tfrac{1}{\sqrt{\gamma_1}},\tfrac{1}{\sqrt{\gamma_2}},\tfrac{1}{\sqrt{\gamma_3}},\tfrac{1}{\sqrt{\gamma_4}}\big]$
that is used when testing.
The resulting ``genie DnCNN'' is specialized to that particular SD vector, and thus not useful in practical situations where the test SD is unknown during training (e.g., in D-GEC). 

The results of our denoiser comparison are presented in \tabref{results_denoiser} using the metrics of PSNR and SSIM \cite{Wang:TIP:04} \textb{along with the respective standard errors (SE)}.
In the first four rows of the table, performance is evaluated for a fixed value of the SD vector $\big[\tfrac{1}{\sqrt{\gamma_1}},\tfrac{1}{\sqrt{\gamma_2}},\tfrac{1}{\sqrt{\gamma_3}},\tfrac{1}{\sqrt{\gamma_4}}\big]$,
while in the last row the results are averaged over subband SDs $\{1/\sqrt{\gamma_\ell}\}_{\ell=1}^4$ drawn independently from a uniform distribution over the interval $[0,50/255]$.
The fourth row corresponds to white Gaussian noise with \textb{a} fixed standard deviation of $10$, while all other rows correspond to colored noise.
The fifth row corresponds to noise that is non-Gaussian in general, but Gaussian when conditioned on $\vec{\gamma}$.
All results in the table represent the average over $500$ different noise realizations.
The results in \tabref{results_denoiser} are summarized as follows.
\begin{itemize}
\item
As expected, white DnCNN performs relatively poorly for all test cases except that in the fourth row, where the testing noise was white, and that in the third row, where the testing noise was lightly colored.
In the fourth row, white DnCNN performs slightly worse than genie DnCNN, which is expected because white DnCNN was trained using white noise with SDs in the range $[0,50/255]$, while genie DnCNN was trained using a white noise with a fixed SD that exactly matches the test noise. 
\item
As expected, genie DnCNN is the best method in the first four rows.
In all of those cases, genie DnCNN is specialized to handle exactly the noise distribution used for the test, and thus is impractical. 
By definition, genie DnCNN is not applicable to the fifth row.
\item
Metzler's DnCNN performs relatively well in the first two rows, but relatively poorly in the second two rows.
We believe that the inconsistency is the result of the fact that the DNN operates in the pixel domain, while the SD side information is given in the wavelet domain and the DNN is given no information about the wavelet transform itself.
\item
The proposed corr+corr outperforms Metzler's DnCNN in all cases and is only $0.3$ to $0.5$~dB away from the genie DnCNN.
This is notable because genie DnCNN gives an (impractical) upper bound on the performance achievable with the chosen architecture and training method.
\end{itemize}

\begin{table}[t]
	\centering
	\caption{Performance comparison of four different DnCNN denoisers for various cases of colored noise}
\blue
	\resizebox{1.0\columnwidth}{!}{
	\begin{tabular}{@{}|c||cc|cc|cc||cc|}\hline
		test standard deviations & \multicolumn{2}{c|}{white DnCNN} & \multicolumn{2}{c|}{Metzler's DnCNN}& \multicolumn{2}{c||}{corr+corr DnCNN} & \multicolumn{2}{c|}{genie DnCNN} \\
		$\big[\frac{1}{\sqrt{\gamma_1}},\frac{1}{\sqrt{\gamma_2}},\frac{1}{\sqrt{\gamma_3}},\frac{1}{\sqrt{\gamma_4}}\big]$  & PSNR $\pm$ SE   & SSIM $\pm$ SE & PSNR $\pm$ SE & SSIM $\pm$ SE & PSNR $\pm$ SE & SSIM $\pm$ SE & PSNR $\pm$ SE & SSIM $\pm$ SE \\ \hline
		{$\big[\frac{48}{255},\frac{47}{255},\frac{6}{255},\frac{19}{255}\big]$}  & 25.36 $\pm$ 0.02 & 0.7328 $\pm$ 0.0013 & 31.23 $\pm$ 0.03 & 0.8783 $\pm$ 0.0006 & 31.69 $\pm$ 0.03 & 0.8899 $\pm$ 0.0005 & 32.12 $\pm$ 0.04 & 0.9012 $\pm$ 0.0005 \\ 
		{$\big[\frac{10}{255},\frac{40}{255},\frac{23}{255},\frac{14}{255}\big]$}   & 32.44 $\pm$ 0.03 & 0.9044 $\pm$ 0.0006 & 34.87 $\pm$ 0.04 & 0.9363 $\pm$ 0.0004 & 35.24 $\pm$ 0.04 & 0.9407 $\pm$ 0.0004 & 35.54 $\pm$ 0.04 & 0.9449 $\pm$ 0.0004 \\ 
		{$\big[\frac{13}{255},\frac{7}{255},\frac{8}{255},\frac{10}{255}\big]$}   & 36.50 $\pm$ 0.03 & 0.9421 $\pm$ 0.0003 & 31.03 $\pm$ 0.03 & 0.9359 $\pm$ 0.0003 & 37.02 $\pm$ 0.03 & 0.9535 $\pm$ 0.0003 & 37.41 $\pm$ 0.03 & 0.9569 $\pm$ 0.0003 \\ 
		{$\big[\frac{10}{255},\frac{10}{255},\frac{10}{255},\frac{10}{255}\big]$} & 37.41 $\pm$ 0.03 & 0.9571 $\pm$ 0.0003 & 31.94 $\pm$ 0.02 & 0.9413 $\pm$ 0.0003 & 37.31 $\pm$ 0.03 & 0.9559 $\pm$ 0.0003 & 37.63 $\pm$ 0.03 & 0.9586 $\pm$ 0.0003 \\ 
		{$\big[0$-$\frac{50}{255},0$-$\frac{50}{255},0$-$\frac{50}{255},0$-$\frac{50}{255}\big]$}  & 31.07 $\pm$ 0.05 & 0.8597 $\pm$ 0.0013 & 33.24 $\pm$ 0.05 & 0.9132 $\pm$ 0.0006 & 34.08 $\pm$ 0.05 & 0.9213 $\pm$ 0.0006 & n/a & n/a \\  \hline
	\end{tabular}}
\color{black}
	\label{tab:results_denoiser}
\end{table}


Code for our corr+corr experiments can be found at \href{https://github.com/Saurav-K-Shastri/corr-plus-corr}{https://github.com/Saurav-K-Shastri/corr-plus-corr}.

\subsection{Example D-GEC behavior in multicoil MRI \textb{with a 2D line mask}} \label{sec:example}

In this section, we demonstrate the typical behavior of D-GEC when applied to multicoil MRI image recovery \textb{with a 2D line mask; experiments with a 2D point mask will be presented in \secref{multi_point}.}
The full details of our multicoil experimental setup are given in \iftoggle{include_app}{\appref{multi}}{Appendix C-A}.
One of our main goals is to demonstrate that D-GEC's denoiser input error behaves as in \eqref{dgec_e2}, i.e., that the error in each wavelet band is white and Gaussian with a predictable variance. 
For the experiments in this section, we used 
the corr+corr denoiser proposed in \secref{corr+corr}, 
a signal-to-noise ratio (SNR) of $40$ dB, and
an acceleration of $R=4$.
Code for our D-GEC experiments can be found at \href{https://github.com/Saurav-K-Shastri/D-GEC}{https://github.com/Saurav-K-Shastri/D-GEC}.

Before discussing our results, there is one peculiarity to multicoil MRI that should be explained.
In practice, both the coil-sensitivity maps $\{\vec{s}_c\}_{c=1}^C$ in $\vec{A}$ from \eqref{A} and the image $\vec{x}\true$ in \eqref{y} are unknown.
The standard recovery approach is to first use an algorithm like ESPIRiT \cite{Uecker:MRM:14} to estimate the coil maps $\{\vec{s}_c\}_{c=1}^C$, then plug the estimated maps into the $\vec{A}$ matrix, and finally solve the inverse problem with the estimated $\vec{A}$ to recover $\vec{x}\true$.
One complication with ESPIRiT is that, in pixel regions where the true image $\vec{x}\true$ is zero or nearly zero (e.g., the outer regions of many MRI images), the ESPIRiT-estimated coil maps can be uniformly zero-valued,
\textb{depending on how ESPIRiT is configured}.
In other words, there \textb{may} exist pixels $n$ such that $[\vec{s}_c]_n=0~\forall c=1\dots C$, which causes the corresponding columns of $\vec{A}$ to be zero.
\textb{In our experiments, we use the default ESPIRiT parameters from the SigPy implementation\footnote{https://sigpy.readthedocs.io/en/latest/generated/sigpy.mri.app.EspiritCalib.html.} and find such zero-valued regions do occur.}
Although the presence of zero-valued columns in $\vec{A}$ might appear to make the inverse problem \eqref{y} more difficult, the (known) coil-map estimates can be exploited as side-information to tell the algorithm which pixels in $\vec{x}\true$ are nearly zero-valued.
\emph{Consequently, in our multicoil experiments, for all algorithms, we set those pixels of the recovered image $\hvec{x}$ to zero wherever the estimated coil maps are uniformly zero.}
In the sequel, we will refer to the pixel region with zero-valued coil map estimates as the ``zero-coil region.'' 

\begin{figure}[t]
    \centering
    \includegraphics[width =\linewidth]{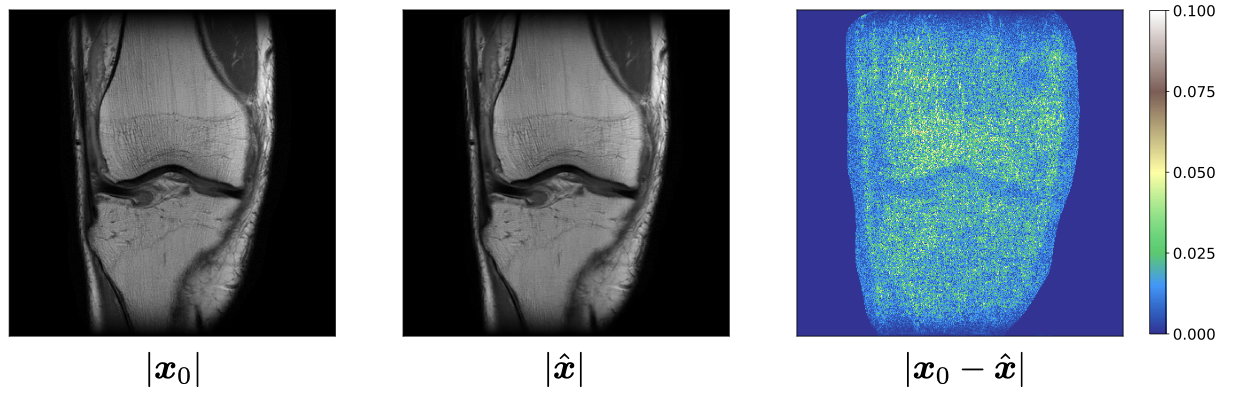}
    \caption{Example multicoil knee image recovery: True image magnitude $|\vec{x}\true|$, D-GEC's recovered image magnitude $|\hvec{x}|$ at iteration $20$, and the error magnitude $|\vec{x}\true - \hvec{x}|$, for $R=4$ and measurement SNR $=40$ dB.}
    \label{fig:GT_Recon_error_knee_R4}
\end{figure}

\begin{figure}[t]
    \centering
    \includegraphics[width = \linewidth]{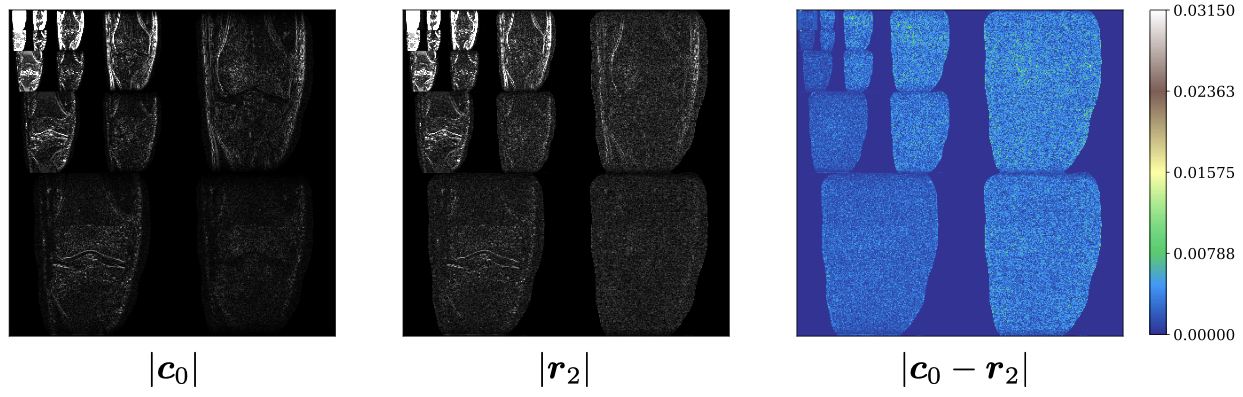}
    \caption{Example multicoil knee image recovery: True wavelet coefficient magnitude $|\vec{c}_0|$, D-GEC's denoiser-input magnitude $|\vec{r}_2|$ at iteration $10$, and the error magnitude $|\vec{c_0} - \vec{r}_2|$, for $R=4$ and measurement SNR $=40$ dB.}
    \label{fig:Wavelets_and_error_knee_R4}
\end{figure}

For a typical MRI knee image, \figref{GT_Recon_error_knee_R4} shows the magnitude $|\vec{x}\true|$ of the true image, D-GEC's recovery $|\hvec{x}|$ after $20$ iterations, and the error magnitude $|\hvec{x}-\vec{x}\true|$.
The error is exactly zero in the previously defined zero-coil region because both $\vec{x}\true$ and $\hvec{x}$ are zero-valued there.
The PSNR $\defn 10\log_{10}[(N \max_n |[\vec{x}\true]_n|^2 )/\|\hvec{x}-\vec{x}\true\|^2]$
and SSIM \cite{Wang:TIP:04} values for this example reconstruction were $36.87$ dB and $0.9397$, respectively. 

\figref{Wavelets_and_error_knee_R4} shows the magnitude $|\vec{c}\true|$ of the corresponding true wavelet coefficients, the magnitude $|\vec{r}_2|$ of the noisy signal entering the D-GEC denoiser at iteration $10$, and the error magnitude $|\vec{r}_2-\vec{c}\true|$. 
The wavelet subbands are visible as the image tiles in these plots.
Here again, we see zero-valued error in the zero-coil region.
As anticipated from \eqref{dgec_e2}, the error maps look like white noise outside the zero-coil region of each wavelet subband, with an error variance that varies across subbands.

To verify the Gaussianity of the wavelet subband errors, \figref{QQ_plots_iter1} shows quantile-quantile (QQ) plots of the real and imaginary parts of the error $\vec{c}\true-\vec{r}_2$ outside the zero-coil region of several wavelet subbands at iteration \textb{$1$, and \figref{QQ_plots_iter10} shows the same at iteration $10$.} 
These QQ-plots suggest that the subband errors are indeed Gaussian \textb{at all iterations}. 

To show that the subband precisions $\vec{\gamma}_2$ predicted by D-GEC match the empirical subband precisions in the error vector $\vec{e}_2$, \figref{gamma_2_evolution} plots the $\ell$th subband SD $1/\sqrt{\gamma_\ell}$ versus iteration, along with the SDs empirically estimated from $\vec{c}\true-\vec{r}_2$, for several subbands $\ell$ and a typical run of the algorithm.  
It can be seen that the predicted SDs are in close agreement with the empirically estimated SDs.

\textb{Finally, to verify that the errors $\vec{c}\true-\vec{r}_2$ are zero-mean in each subband of each validation image, we performed a t-test \cite{Walpole:Book:16} using a significance level of $\alpha=0.05$ (i.e., if the errors were truly zero mean then the test would fail with probability $\alpha$).
At the first iteration, we ran a total of $208$ tests (one for each of the $13$ subbands in each of the $16$ knee validation images at $R=4$ and SNR $=40$ dB) and found that $11$ tests rejected the zero-mean hypothesis, which is consistent with $\alpha=0.05$ since $11/208=0.0529\approx 0.05$.
At the $10$th iteration, $12$ tests rejected the zero-mean hypothesis, which is again consistent with $\alpha=0.05$. 
}

\begin{figure}[t]
    \centering
    \includegraphics[width = 0.5\linewidth]{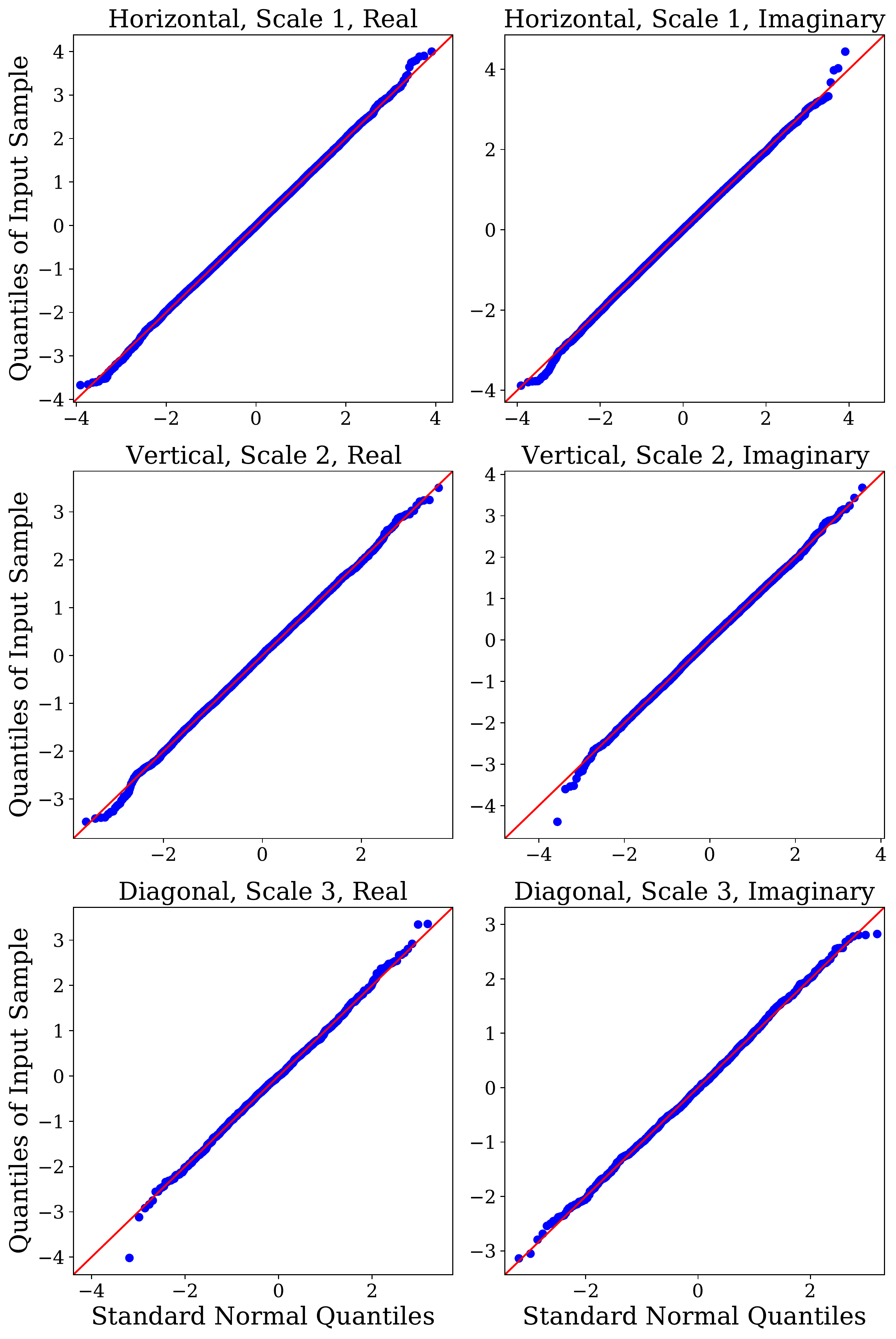}
    \caption{QQ-plots of the real and imaginary parts of D-GEC's subband errors $\vec{c}\true-\vec{r}_2$ at iteration $1$.}
    \label{fig:QQ_plots_iter1}
\end{figure}

\begin{figure}[t]
    \centering
    \includegraphics[width = 0.5\linewidth]{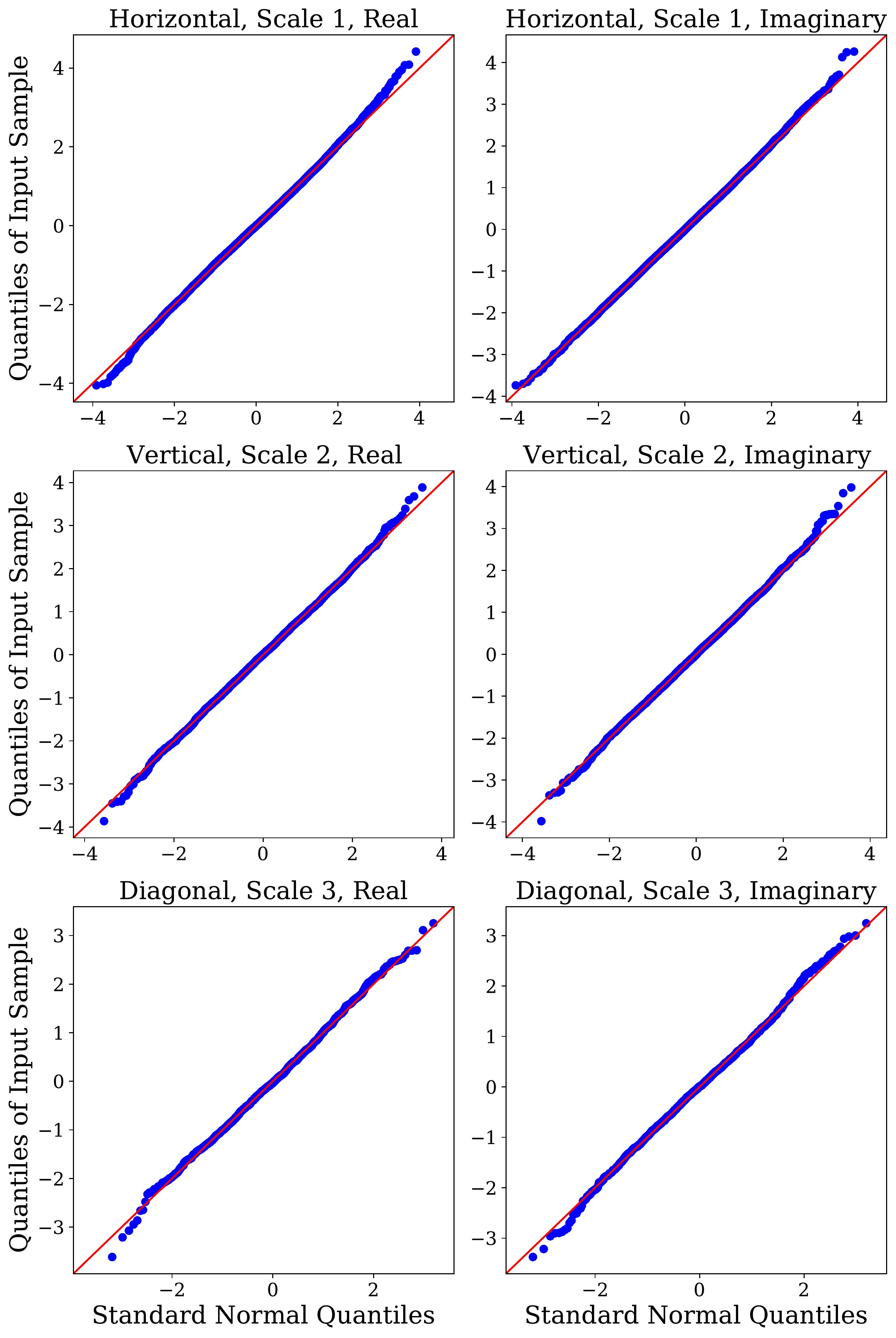}
    \caption{QQ-plots of the real and imaginary parts of D-GEC's subband errors $\vec{c}\true-\vec{r}_2$ at iteration $10$.}
    \label{fig:QQ_plots_iter10}
\end{figure}

\begin{figure}[t]
    \centering
    \includegraphics[width = \linewidth]{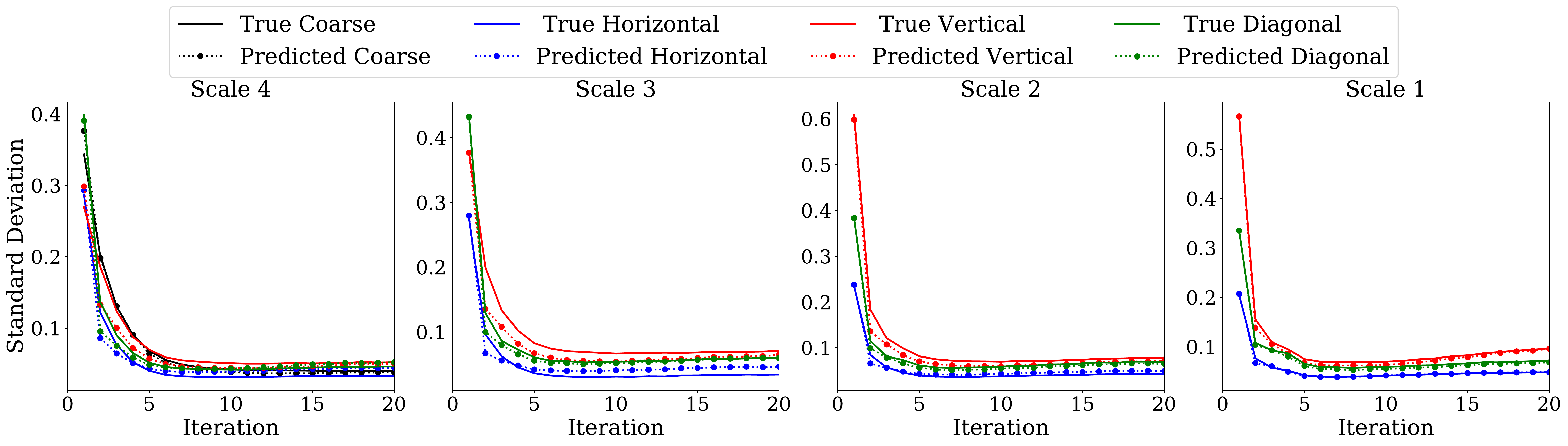}
    \caption{Evolution of D-GEC's predicted subband SDs ($1/\sqrt{\gamma_\ell}$) and empirically estimated subband SDs (from $\vec{c}\true - \vec{r}_2$) for several subbands $\ell$ over 20 iterations.}
    \label{fig:gamma_2_evolution}
\end{figure}

\subsection{Multicoil MRI algorithm comparison \textb{with a 2D point mask}} \label{sec:multi_point}

In this section, we compare the performance of D-GEC to \textb{two state-of-the-art algorithms for multicoil MRI image recovery: P-VDAMP \cite{Millard:22} and PnP-PDS \cite{Ono:SPL:17}.
We use 2D point masks in this section out of fairness to P-VDAMP, which was designed around 2D point masks. 
Multicoil experiments with 2D line masks are presented in \secref{multi_line}, and single-coil experiments are presented in \secref{single}.}
We examine two acceleration rates, $R=4$ and $R=8$, and several measurement SNRs between $20$ and $45$~dB.
As before, we quantify recovery performance using PSNR and SSIM. 
For this section, we used both knee and brain fastMRI data. 
The details of the experimental setup are given in \iftoggle{include_app}{\appref{multi}}{Appendix C-A}.

For P-VDAMP, we ran the \textb{authors'} code from \cite{Millard:22} under its default settings.
For PnP-PDS, we used a bias-free DnCNN \cite{Mohan:ICLR:20} denoiser trained to minimize \textb{$\ell_2$} loss when removing WGN with an SD uniformly distributed in the interval $[0,55/255]$.
This bias-free network is known to perform very well over a wide SD range, and so there is no advantage in training multiple denoisers over different SNR ranges \cite{Mohan:ICLR:20}.  
Because PnP-PDS \textb{performance} strongly depends on the chosen penalty parameter \textb{and number of PDS iterations, we separately tuned these parameters for every combination of} measurement SNR and acceleration rate to maximize PSNR on the training set. 
For D-GEC, we used a Haar wavelet transform of depth $D=4$, which yields $L=13$ subbands, and a corr+corr bias-free DnCNN denoiser; see \iftoggle{include_app}{\appref{multi}}{Appendix C-A} for additional details.
For all algorithms, we set the image estimate to zero in the zero-coil region.

For each acceleration rate $R$ and SNR under test, we ran all three algorithms on all images in the brain and knee testing sets. 
We then computed the average PSNR and SSIM values across those images and summarized the results in \figref{SNRvsMetric_all}, \textb{using error bars to show plus/minus one standard error}. 
The figure shows that D-GEC significantly outperformed the other algorithms in all metrics at all combinations of $R$ and measurement SNR.

\begin{figure}[t]
    \centering
    \includegraphics[width = \linewidth]{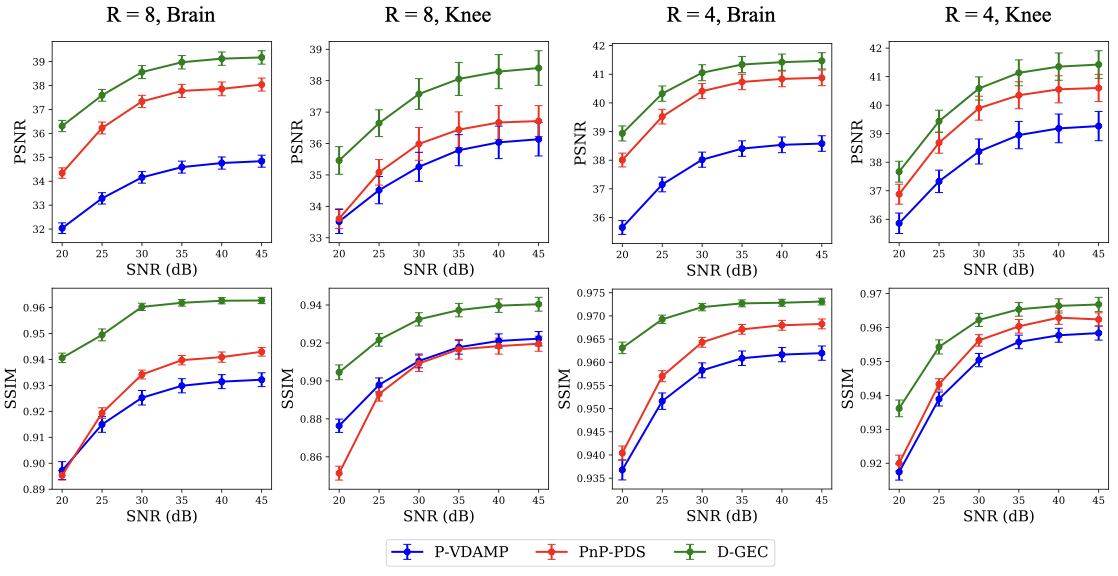}
    \caption{Average PSNR and SSIM versus measurement SNR for P-VDAMP, PnP-PDS, and D-GEC.}
    \label{fig:SNRvsMetric_all}
\end{figure}

\Figref{multi_coil_recon_comparison} shows image recoveries and error images for a typical fastMRI brain image at acceleration $R=4$ and measurement SNR $=35$~dB. 
In this case, D-GEC outperformed the P-VDAMP and PnP-PDS algorithms in PSNR by $2.6$ and $0.76$ dB, respectively.
Furthermore, D-GEC's error image looks the least structured.
Looking at the details of the zoomed plots, we see that D-GEC is able to reconstruct certain fine details better than its competitors. 

\Figref{PSNRvsIter_multicoil} shows PSNR versus iteration for the three algorithms at $R=4$ and SNR $=20$~dB. 
The PSNR values shown are the average over all $16$ test images from the brain MRI dataset. 
The plot shows P-VDAMP, D-GEC, and PnP-PDS taking about $7$, $8$, and $25$ iterations to converge, respectively.
If we measure the number of iterations taken to reach $35$~dB SNR, then D-GEC, PnP-PDS, and P-VDAMP take about $3$, $5$, and $7$ iterations, respectively. 

\begin{figure}[t]
    \centering
    \includegraphics[width = \linewidth]{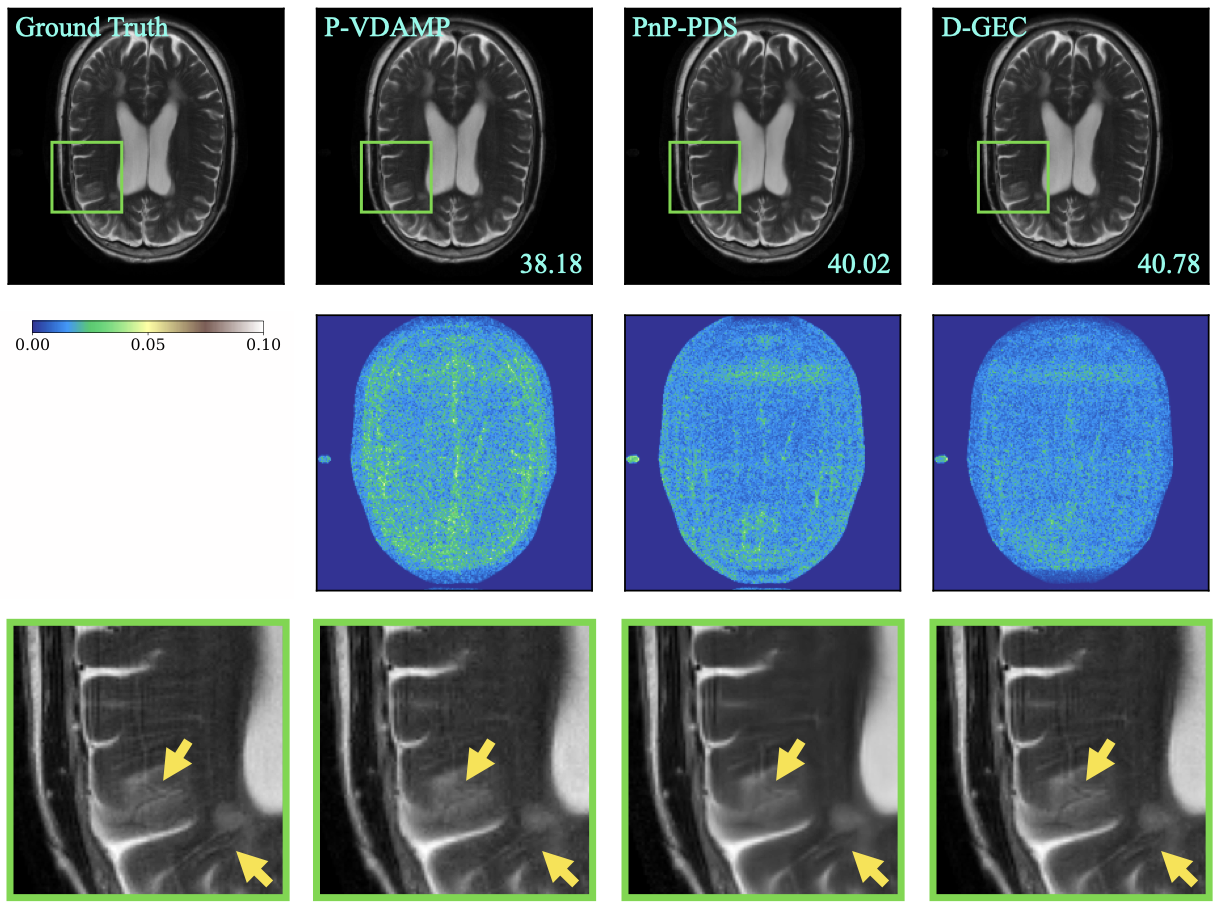}
    \caption{Example multicoil MRI image recoveries and error images at $R=4$ and SNR $=35$~dB. The number printed on each recovered image shows its PSNR. The bottom row is a zoomed in version of the green square in the top row.  This figure is best viewed in electronic form. }
    \label{fig:multi_coil_recon_comparison}
\end{figure}

\begin{figure}[t]
    \centering
    \includegraphics[width = 0.5\linewidth]{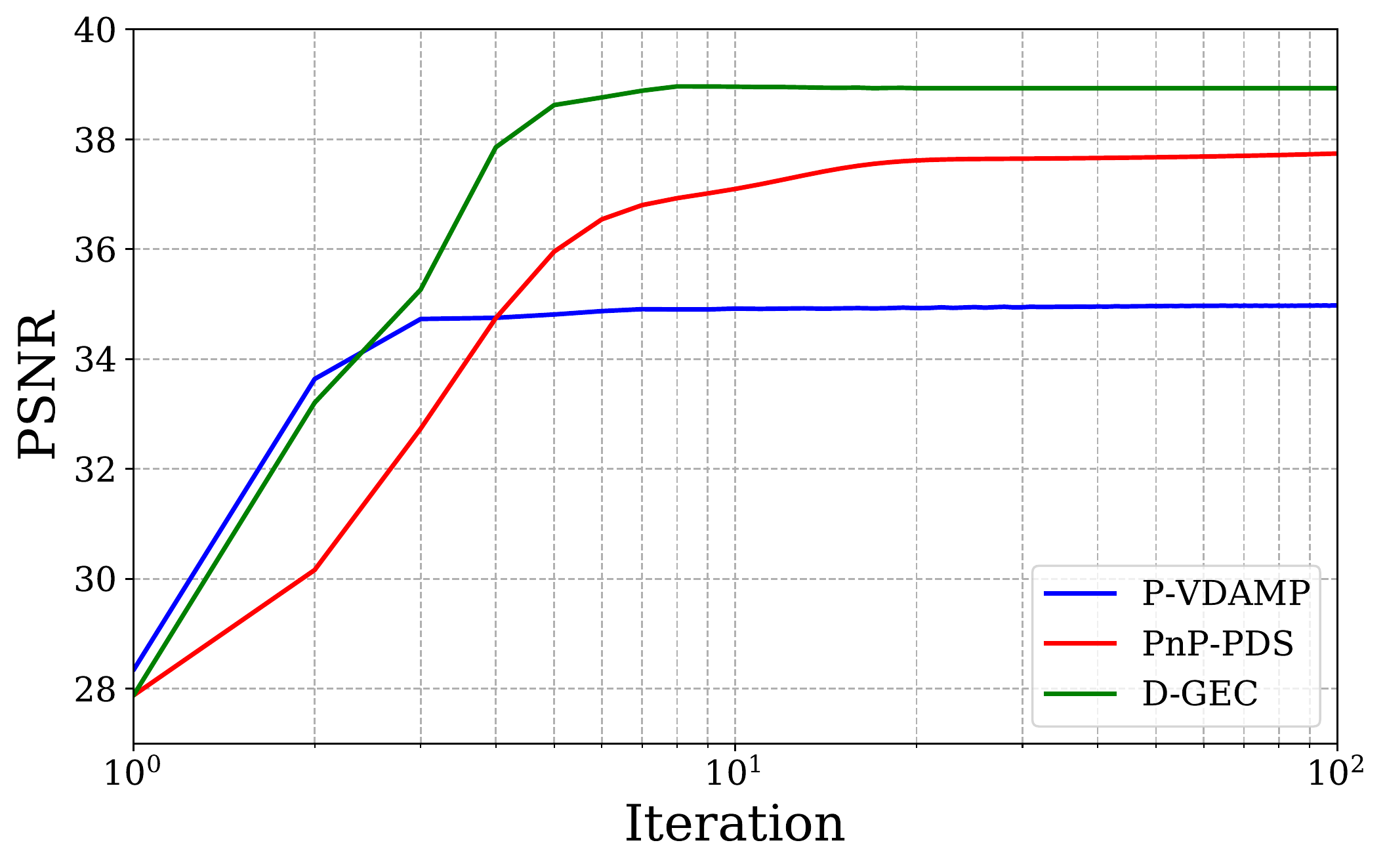}
    \caption{PSNR versus iterations for multicoil brain MRI recovery at $R=4$ and SNR $=20$~dB. PSNR was averaged over the $16$ test images.}
    \label{fig:PSNRvsIter_multicoil}
\end{figure}

\blue
\subsection{Multicoil MRI algorithm comparison with a 2D line mask} \label{sec:multi_line}

In this section, we compare the performance of D-GEC to that of P-VDAMP \cite{Millard:22} and PnP-PDS \cite{Ono:SPL:17} when using a 2D line mask.
We examine acceleration rates $R=4$ and $R=8$, and a measurement SNR of $40$~dB, on the fastMRI brain and knee datasets.
With the exception of the sampling mask, the experimental setup was identical to that in \secref{multi_point}.
Although \cite{Millard:22} states that P-VDAMP is not intended to be used for ``purely 2D acquisitions'' like that associated with a 2D line mask, we show P-VDAMP performance for completeness.
To run P-VDAMP, we gave it a 2D sampling density that was uniform along the fully sampled dimension and proportional to the 1D sampling density along the subsampled dimension (recall Figs.~\ref{fig:mask}(c)-(d)).

\tabref{results_MC} shows PSNR and SSIM averaged over the test images with the corresponding standard errors. 
There it can be seen that D-GEC significantly outperformed the other techniques on both datasets at both acceleration rates. 
For example, D-GEC outperformed its closest competitor, PnP-PDS, by $2.54$ and $1.32$ dB at $R=4$ and $R=8$, respectively, on the knee data.

\begin{table}[t]
\blue
	\centering
	\caption{\textb{Multicoil 2D line-mask results at SNR $=40$ dB averaged over all test images.}}
	\resizebox{\columnwidth}{!}{
	\begin{tabular}{@{}|c|cc|cc|cc|cc|}\hline
	    & \multicolumn{4}{c|}{Knee} & \multicolumn{4}{c|}{Brain}\\ \hline
		& \multicolumn{2}{c|}{$R = 4$} & \multicolumn{2}{c|}{$R = 8$} & \multicolumn{2}{c|}{$R = 4$} & \multicolumn{2}{c|}{$R = 8$} \\
		method  & PSNR $\pm$ SE  & SSIM $\pm$ SE & PSNR $\pm$ SE  & SSIM $\pm$ SE & PSNR $\pm$ SE  & SSIM $\pm$ SE & PSNR $\pm$ SE  & SSIM $\pm$ SE \\ \hline
		P-VDAMP \cite{Millard:22} & 33.84 $\pm$ 0.40 & 0.9018 $\pm$ 0.0036 &  20.34  $\pm$  0.46 & 0.5614  $\pm$ 0.0051 &  30.30  $\pm$  0.16 & 0.8847  $\pm$ 0.0021 &  13.51  $\pm$  0.26 & 0.4763  $\pm$ 0.0069 \\ 
		PnP-PDS \cite{Ono:SPL:17} & 36.28 $\pm$ 0.38  & 0.9204 $\pm$ 0.0028 & 32.34 $\pm$ 0.32  & 0.8556 $\pm$ 0.0040 &  38.07  $\pm$  0.23 & 0.9501  $\pm$ 0.0016 &  28.97  $\pm$  0.13 & 0.8269  $\pm$ 0.0031 \\
		D-GEC (proposed) & \textbf{38.82} $\pm$ 0.50  & \textbf{0.9504} $\pm$ 0.0023 &  \textbf{33.66} $\pm$ 0.28  & \textbf{0.8893} $\pm$ 0.0028 &  \textbf{39.04}  $\pm$  0.29 & \textbf{0.9631}  $\pm$ 0.0013 &  \textbf{30.61}  $\pm$  0.19 & \textbf{0.9015}  $\pm$ 0.0031  \\ \hline
	\end{tabular}}
	\label{tab:results_MC}
\end{table}

\color{black}
\subsection{Single-coil MRI algorithm comparison \textb{with a 2D point mask}} \label{sec:single}

In this section we compare the performance of D-GEC to several other recently proposed algorithms for single-coil MRI recovery \textb{using a 2D point mask}.
We examine two acceleration rates, $R=4$ and $R=8$, and a measurement SNR of $45$~dB.
For this section, we used the Stanford 2D FSE dataset \cite{Ong:ISMRM:18} with the test images in \figref{test_images}.
The details of the experimental setup are reported in \iftoggle{include_app}{\appref{single}}{Appendix C-B}.

We compared our proposed D-GEC algorithm to D-AMP-MRI \cite{Eksioglu:JIS:18}, VDAMP \cite{Millard:OJSP:20}, D-VDAMP \cite{Metzler:ICASSP:21}, and PnP-PDS \cite{Ono:SPL:17}. 
\textb{We used a 2D point mask out of fairness to VDAMP and D-VDAMP, which were designed around 2D point masks.
For VDAMP and D-VDAMP, we ran the authors' implementations} at their default settings. 
For D-AMP-MRI and PnP-PDS, we used a bias-free DnCNN \cite{Mohan:ICLR:20} denoiser trained to minimize the \textb{$\ell_2$} loss when removing WGN with SDs uniformly distributed in the interval $[0,55/255]$. 
This bias-free network is known to perform very well over a wide SD range, and so there is no advantage in training multiple denoisers over different SNR ranges \cite{Mohan:ICLR:20}. 
We ran the D-AMP-MRI and PnP-PDS algorithms for $50$ and $300$ iterations, respectively.  
Because the PnP fixed-points strongly depend on the chosen penalty parameter, we carefully tuned the PnP-PDS parameter at each acceleration rate $R$ to maximize PSNR on the validation set.
For D-GEC, we used a Haar wavelet transform of depth $D=4$, which yields $L=13$ subbands, and a corr+corr bias-free DnCNN denoiser; see \iftoggle{include_app}{\appref{single}}{Appendix C-B} for additional details.  

\tabref{results_SC} shows PSNR and SSIM averaged over the $10$ test images \textb{with the corresponding standard errors}.
There it can be seen that D-GEC significantly outperformed the other techniques at both tested acceleration rates. 
For example, D-GEC outperformed its closest competitor, PnP-PDS, by $1.81$ and $0.87$ dB at $R=4$ and $R=8$, respectively.

\begin{table}[t]
	\centering
	\caption{Single-coil image recovery results averaged over the ten test images.}
	\resizebox{0.75\columnwidth}{!}{
	\begin{tabular}{@{}|c|cc|cc|}\hline
		& \multicolumn{2}{c|}{$R = 4$} & \multicolumn{2}{c|}{$R = 8$} \\
		method  & PSNR \textb{$\pm$ SE}  & SSIM \textb{$\pm$ SE} & PSNR \textb{$\pm$ SE}  & SSIM \textb{$\pm$ SE} \\ \hline
		D-AMP-MRI \cite{Eksioglu:JIS:18} & 33.28 $\pm$ 4.62    & 0.7789 $\pm$ 0.0900 &  25.83 $\pm$ 4.33  & 0.7252 $\pm$ 0.1214  \\ 
		VDAMP \cite{Millard:OJSP:20} & 33.10 $\pm$ 1.30  & 0.8650 $\pm$ 0.0243  &  28.47  $\pm$ 0.96    & 0.7378 $\pm$ 0.0313  \\ 
		D-VDAMP \cite{Metzler:ICASSP:21} & 42.57 $\pm$ 1.48 & 0.9731 $\pm$ 0.0089 &  35.18  $\pm$  1.93 & 0.9023  $\pm$ 0.0248 \\ 
		PnP-PDS \cite{Ono:SPL:17} & 43.36 $\pm$ 1.60  & 0.9787 $\pm$ 0.0076 &  38.10 $\pm$ 1.75    & 0.9527 $\pm$ 0.0158\\
		D-GEC (proposed) & \textbf{45.17} $\pm$ 1.62  & \textbf{0.9824} $\pm$ 0.0066  &  \textbf{38.97} $\pm$ 1.76  & \textbf{0.9570} $\pm$ 0.0132  \\
                \hline
	\end{tabular}}
	\label{tab:results_SC}
\end{table}

\Figref{PSNRvsIter_single} shows PSNR versus iteration for several algorithms at $R=4$ and SNR $=45$~dB.
The PSNR value shown is the average over all $10$ test images in \figref{test_images}.
Two versions of D-VDAMP are shown in \figref{PSNRvsIter_single}: the standard version from \cite{Metzler:ICASSP:21}, which includes early stopping, and a modified version without early stopping.  
The importance of early stopping is clear from the figure.
The figure also shows that, for this single-coil dataset, D-GEC took more iterations to converge than the other algorithms but yielded a larger value of PSNR at convergence.
In the multicoil case in \figref{PSNRvsIter_multicoil}, D-GEC took an order-of-magnitude fewer iterations to converge.

\begin{figure}[t]
    \centering
    \includegraphics[width = 0.5\linewidth]{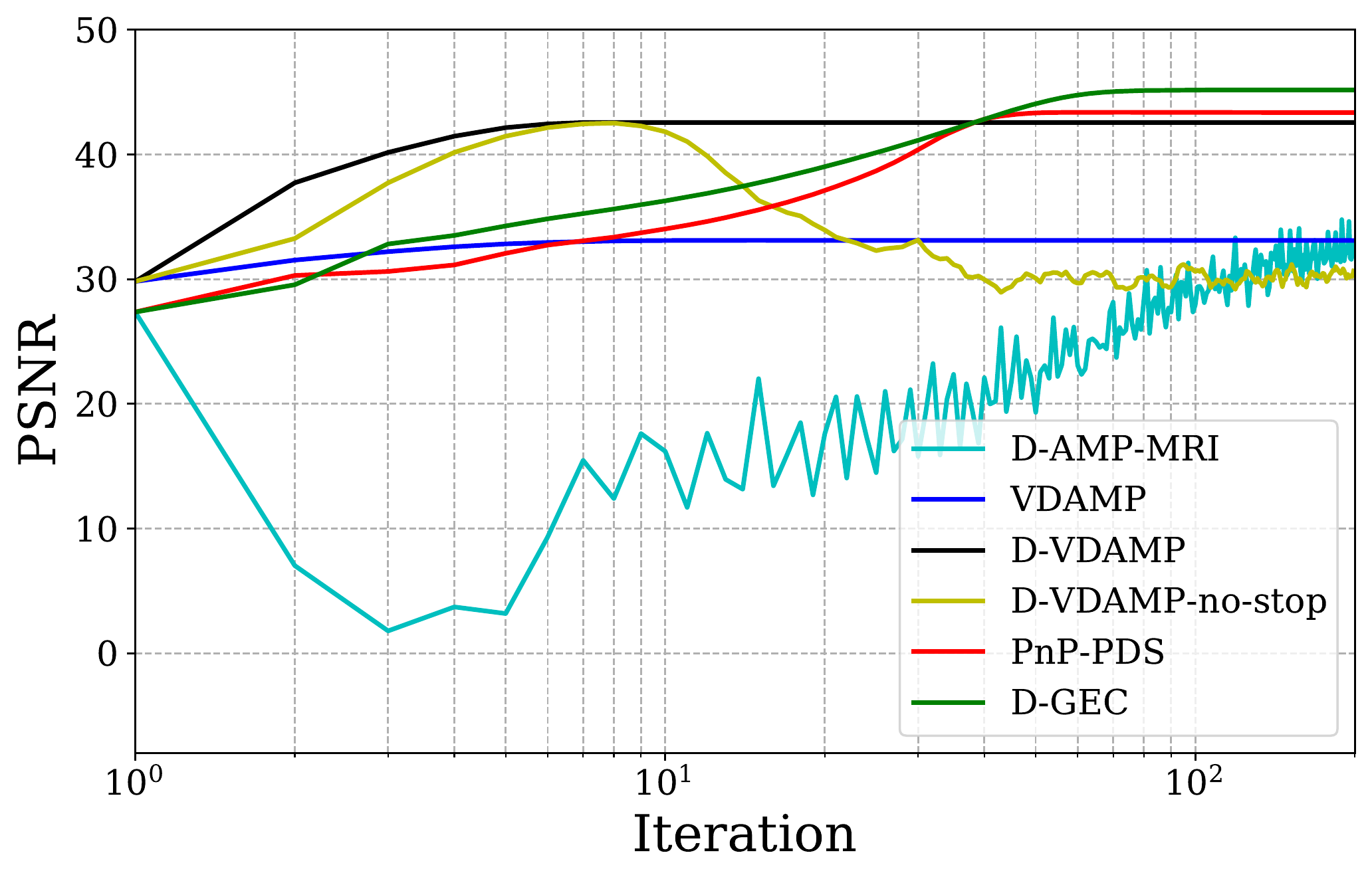}
    \caption{PSNR versus iterations for single-coil MRI recovery at $R=4$ and SNR $=45$~dB. PSNR was averaged over the $10$ test images in \figref{test_images}.}
    \label{fig:PSNRvsIter_single}
\end{figure}

\Figref{single_coil_recon_comparison} shows image recoveries for a typical Stanford 2D FSE MRI image at $R=4$ and measurement SNR $=45$~dB. 
For this experiment, D-GEC significantly outperformed the competing algorithms in PSNR, and its error image looks the least structured.
Also, the zoomed subplots show that D-GEC recovered fine details in the true image that are missed by its competitors.

\begin{sidewaysfigure}[p]
    \centering
    \includegraphics[width = \linewidth]{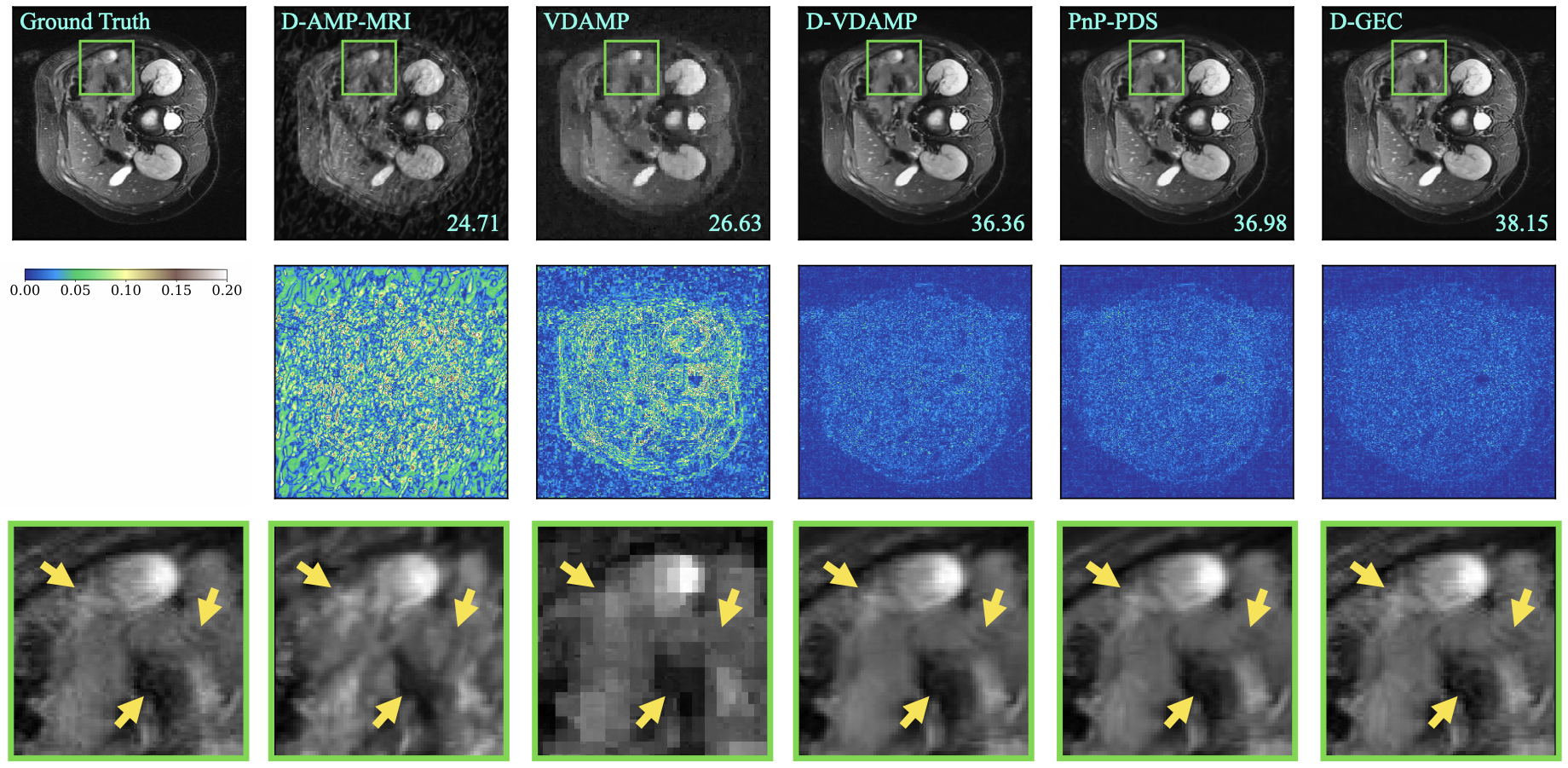}
    \caption{Example single-coil MRI image recoveries and error images at $R=4$ and SNR $=45$ dB.  The number printed on each recovered image shows its PSNR.  The bottom row is a zoomed in version of the green square in the top row. This figure is best viewed in electronic form.}
    \label{fig:single_coil_recon_comparison}
\end{sidewaysfigure}


\section{Conclusion}

\textb{PnP algorithms} 
require relatively few training images and are insensitive to deviations in the forward model $\vec{A}$ and measurement noise statistics between training and test.
However, PnP can be improved, because the denoisers typically used for PnP are trained to remove white Gaussian noise, whereas the denoiser input errors encountered in PnP are typically non-white and non-Gaussian.
In this paper, we proposed a new PnP algorithm, called Denoising Generalized Expectation-Consistent (D-GEC) approximation, to address this shortcoming for Fourier-structured $\vec{A}$ and Gaussian measurement noise.
In particular, D-GEC is designed to make the denoiser input error white and Gaussian within each wavelet subband with a predictable variance.
We then proposed a new DNN denoiser that is capable of exploiting the knowledge of those subband error variances.
Our ``corr+corr'' denoiser takes in a signal corrupted by correlated Gaussian noise, as well as independent realization(s) of the same correlated noise.
It then learns how to extract the statistics of the provided noise and then use them productively for denoising the signal. 
Numerical experiments with single- and multicoil MRI image recovery demonstrate that D-GEC does indeed provide the denoiser with subband errors that are white and Gaussian with a predictable variance.
Furthermore, the experiments demonstrate improved recovery accuracy relative to existing state-of-the-art PnP methods for MRI, \textb{especially with practical 2D line sampling masks}.
More work is needed to understand the theoretical properties of the proposed D-GEC and corr+corr denoisers.

\clearpage
\bibliographystyle{IEEEtran}
\bibliography{macros_abbrev,phase,books,misc,comm,multicarrier,sparse,machine,mri}

\iftoggle{include_app}{
\appendices

\section{EC/VAMP error recursion} \label{app:recursion}

In this appendix, we establish the error iteration
\begin{align}
\ebf_2
&= \Vbf\vec{D}\Vbf\herm \ebf_1 + \ubf 
\label{eq:ebf2_app}.
\end{align}
To begin, we write the estimation function $\vec{f}_1$ 
\iftoggle{include_app}{from \eqref{f1_ec_awgn}}{} as
\begin{align}
\lefteqn{
\vec{f}_1(\vec{r}_1;\gamma_1)
= \left(\gamma_w\vec{A}\herm\vec{A}+\gamma_1\vec{I}\right)^{-1}\left(\gamma_w\vec{A}\herm\vec{y} + \gamma_1\vec{r}_1\right) } \\
&= \rbf_1 + \left(\gamma_w\Abf\herm\Abf + \gamma_1\Ibf_N \right)^{-1} \left(\gamma_w\Abf\herm\ybf-\gamma_w\Abf\herm\Abf\rbf_1 \right) \\
&= \rbf_1 + \gamma_w \left( \Cbf + \gamma_1 \Ibf_N \right)^{-1} \Abf\herm (\ybf-\Abf\rbf_1)
\end{align}
for 
\begin{align}
\Cbf 
&\defn \gamma_w\Abf\herm\Abf 
\label{eq:Cbf} 
= \Vbf\Lambdabf\Vbf\herm .
\end{align}
The right side of \eqref{Cbf} is an eigendecomposition where $\Vbf\Vbf\herm=\Vbf\herm\Vbf=\Ibf$
and $\vec{\Lambda}=\Diag([\lambda_1,\dots,\lambda_N])$ is real-valued. 
Note also that $\Vbf$ is the right singular vector matrix of $\vec{A}$.
Using this eigendecomposition, we can write
\begin{align}
\lefteqn{
\tr(\nabla \vec{f}_1(\vec{r}_1;\gamma_1))
= \tr( \vec{I} - (\Cbf + \gamma_1 \Ibf_N)^{-1}\Cbf ) }\\
&= \tr( \vec{I} - (\Vbf\Lambdabf\Vbf\herm + \gamma_1 \Ibf_N)^{-1}\Vbf\Lambdabf\Vbf\herm ) \\
&= \tr( \vec{I} - (\Lambdabf + \gamma_1 \Ibf_N)^{-1}\Lambdabf ) \\
&= N - \sum_{n=1}^N \frac{\lambda_n}{\lambda_n+\gamma_1} \\
&= N(1-\alpha) \text{~~for~~} \alpha \defn \frac{1}{N}\sum_{n=1}^N \frac{\lambda_n}{\lambda_n+\gamma_1}
\label{eq:alpha} .
\end{align}
Thus, 
\iftoggle{include_app}{%
lines~\ref{line:ec_x1}-\ref{line:ec_eta1} of \algref{ec}
}{%
lines~4-5 of Alg.~1
}%
can be written as
\begin{align}
\hvec{x}_1
&= \rbf_1 + \gamma_w \left( \Cbf + \gamma_1 \Ibf_N \right)^{-1} \Abf\herm (\ybf-\Abf\rbf_1)
\label{eq:xbfhat1}\\
\eta_1 
&= \frac{\gamma_1 N}{\tr(\nabla \vec{f}_1(\vec{r}_1;\gamma_1))}
= \frac{\gamma_1}{1-\alpha} 
\label{eq:eta1inv} 
\end{align}
and 
\iftoggle{include_app}{%
lines~\ref{line:ec_r2}--\ref{line:ec_gam2}
}{%
lines~7-8
}%
as
\begin{align}
\gamma_2
&= \eta_1 - \gamma_1 
= \gamma_1 \left(\frac{1}{1-\alpha}-1\right) 
= \gamma_1 \frac{\alpha}{1-\alpha} \\
\rbf_2
&= \frac{\eta_1\hvec{x}_1 - \gamma_1\vec{r}_1}{\gamma_2} 
= \frac{1}{\alpha}\xbfhat_1 - \frac{1-\alpha}{\alpha}\rbf_1
\label{eq:rbf2} .
\end{align}
Plugging \eqref{xbfhat1} into \eqref{rbf2}, we get
\begin{equation}
\rbf_2
= \rbf_1 + \frac{\gamma_w}{\alpha}\left( \Cbf + \gamma_1\Ibf_N \right)^{-1} \Abf\herm (\ybf-\Abf\rbf_1)
\label{eq:rbf1b} .
\end{equation}

Next, we express \eqref{rbf1b} in terms of the error vectors $\ebf_i\defn \rbf_i-\xbf\true$ for $i=1,2$.
Subtracting $\xbf\true$ from both sides of \eqref{rbf1b} and applying $\vec{y}=\vec{Ax}\true+\vec{w}$ 
\iftoggle{include_app}{from \eqref{y}}{} 
and the definition of $\vec{C}$ from \eqref{Cbf}, we get
\begin{align}
\ebf_2
&= \ebf_1 + \frac{\gamma_w}{\alpha}\left( \Cbf + \gamma_1\Ibf_N \right)^{-1} \Abf\herm (\Abf\xbf\true+\wbf-\Abf\rbf_1) \nonumber \\
&= \ebf_1 - \frac{1}{\alpha}\left( \Cbf + \gamma_1\Ibf_N \right)^{-1} \Cbf\ebf_1 + \ubf 
\label{eq:ebf2a} \\
&= \ebf_1 - \frac{1}{\alpha}\Vbf\left( \Lambdabf + \gamma_1\Ibf_N \right)^{-1} \Lambdabf\Vbf\herm \ebf_1 + \ubf \\
&= \Vbf\vec{D}\Vbf\herm \ebf_1 + \ubf 
\label{eq:ebf2b} ,
\end{align}
where 
\begin{align}
\ubf
&\defn \frac{\gamma_w}{\alpha}\big( \Cbf + \gamma_1\Ibf_N \big)^{-1} \Abf\herm \wbf
\label{eq:ubf}\\
\vec{D} 
&\defn \vec{I}_N - \frac{1}{\alpha}\left( \Lambdabf + \gamma_1\Ibf_N \right)^{-1} \Lambdabf
\label{eq:Dbf} 
\end{align}
Notice that $\tr(\vec{D})=0$ due to the definition of $\alpha$ in \eqref{alpha}.

\section{EC/VAMP error analysis} \label{app:ebf2}

We start with the fact \cite{Collins:JMP:09} that, for any $N\geq 2$, the elements $v_{nj}$ of uniformly distributed orthogonal $\vec{V}\in\Real^{N\times N}$ obey
\begin{subequations}\label{eq:weingarten}
\begin{align}
\Exp(v_{nj}) 
&= 0 
\label{eq:Ev}\\
\Exp(v_{nj} v_{mk})
&= \tfrac{1}{N}\delta_{n-m}\delta_{j-k} 
\label{eq:Ev2}\\
\Exp(v_{nj}^2 v_{mk}^2)
&= \begin{cases}
\frac{3}{N(N+2)} & n=m ~\&~ j=k\\
\frac{1}{N(N+2)} & n=m ~\&~ j\neq k\\
\frac{1}{N(N+2)} & n\neq m ~\&~ j=k\\
\frac{N+1}{N(N+2)(N-1)} & n\neq m ~\&~ j\neq k 
\end{cases}
\label{eq:Ev4} ,
\end{align}
\end{subequations}
where $\delta_n$ is the Kronecker delta (i.e., $\delta_0=1$ and $\delta_n\big|_{n\neq 0}=0$).
Equations~\eqref{weingarten} will be used to establish the following lemma.

\begin{lemma} \label{lem:VDVe}
Suppose that $\fbf = \Vbf\Diag(\dbf)\Vbf\tran \ebf \in \Real^N$ where
$\dbf$ is deterministic with elements obeying 
$\sum_{j=1}^N d_j=0$ and 
$\mc{D}\defn\lim_{N\rightarrow\infty}\frac{1}{N}\sum_{j=1}^N d_j^2 < \infty$; 
$\ebf$ is random with elements of finite mean and variance obeying
$\varepsilon\defn\lim_{N\rightarrow\infty}\frac{1}{N}\sum_{j=1}^N e_{j}^2 < \infty$; 
and
$\Vbf$ is uniformly distributed over the set of orthogonal matrices and independent of $\ebf$ up to the fourth moment, i.e.,
$\Exp(v_{nj}v_{mk}v_{n'j'}v_{m'k'}|\ebf)=\Exp(v_{nj}v_{mk}v_{n'j'}v_{m'k'})$.
%
Then, as $N\rightarrow\infty$,
\begin{align}
\Exp(\fbf|\ebf)
&= \zero 
\label{eq:Ef} \\
\Cov(\fbf|\ebf)
&= \varepsilon \mc{D} \Ibf_N 
\label{eq:Eff} .
\end{align}
\end{lemma}

\begin{proof}
Writing the $n$th element of $\fbf$ as
\begin{align}
f_n
&= \sum_{j=1}^N v_{nj} d_j \sum_{k=1}^N v_{kj} e_{k}
\label{eq:fn}
\end{align}
we can establish \eqref{Ef} via
\begin{align}
\Exp(f_n|\ebf)
&= \sum_{j=1}^N \sum_{k=1}^N d_j e_{k} \Exp( v_{nj} v_{kj} | \ebf) \\
&\stackrel{(a)}{=} \sum_{j=1}^N \sum_{k=1}^N d_j e_{k} \delta_{n-k} \frac{1}{N}\\ 
&= e_{n} \frac{1}{N} \sum_{j=1}^N d_j 
\stackrel{(b)}{=} 0~\forall n ,
\end{align}
where (a) used \eqref{Ev2} and the assumed independence of $\Vbf$ and $\ebf$ and (b) used $\sum_j d_j=0$.

To establish \eqref{Eff}, we begin by using \eqref{fn} and the assumed independence of $\Vbf$ and $\ebf$ to write
\begin{align}
\Exp(f_n^2|\ebf)
&= \sum_{j}\sum_{k}\sum_{j'}\sum_{k'} d_j d_{j'} e_{k} e_{k'} 
        \Exp( v_{nj} v_{kj} v_{nj'} v_{k'j'}) .
\end{align}
When $k=n$, the expectation will vanish unless $k'=n$, and 
when $k\neq n$, the expectation will vanish unless $k'=k$ and $j'=j$.
Thus we have
\begin{align}
\lefteqn{ 
\Exp(f_n^2|\ebf) 
}\nonumber\\
&= e_{n}^2 \sum_{j}\sum_{j'} d_j d_{j'} \Exp( v_{nj}^2 v_{nj'}^2) 
   + \sum_{k\neq n}\sum_{j} d_j^2 k^2 \Exp( v_{nj}^2 v_{kj}^2) \\
&= e_{n}^2 \sum_{j} d_j^2 \Exp( v_{nj}^4 ) 
   + e_{n}^2 \sum_{j}\sum_{j'\neq j} d_j d_{j'} \Exp( v_{nj}^2 v_{nj'}^2) 
\nonumber\\&\quad
   + \sum_{k\neq n}\sum_{j} d_j^2 e_{k}^2 \Exp( v_{nj}^2 v_{kj}^2) \\
&\stackrel{(a)}{=} 
   \frac{3 e_{n}^2}{N(N+2)} \sum_j d_j^2
   + \frac{e_{n}^2}{N(N+2)} \sum_{j} d_j \sum_{j'\neq j} d_{j'}  
\nonumber\\&\quad
   + \frac{1}{N(N+2)} \sum_j d_j^2 \sum_{k\neq n} e_{k}^2 \\
&\stackrel{(b)}{=}
   \frac{e_{n}^2}{N+2} \big(\frac{1}{N}\sum_j d_j^2\big)
   + \frac{N}{N+2} \big(\frac{1}{N}\sum_j d_j^2\big)
   \big( \frac{1}{N}\sum_k e_{k}^2 \big) \\
&\stackrel{N\rightarrow\infty}{=} \mc{D} \varepsilon ,
\end{align}
where (a) used \eqref{Ev4} and where (b) used 
$\sum_{j'\neq j} d_{j'} = (\sum_{j'}d_{j'})-d_j = -d_j$ and
$\sum_{k\neq n} e_{k}^2 = \|\ebf\|^2-e_{n}^2$.
The limit as $N\rightarrow\infty$ follows from the definitions 
of $\mc{D}$ and $\varepsilon$, and the fact that 
$\lim_{N\rightarrow\infty} e_{n}^2/N=0$ due to the 
finite mean and variance of $e_{n}$.
Thus we have established the diagonal terms in \eqref{Eff}.

The off-diagonal terms in \eqref{Eff} follow from analyzing
\begin{align}
\lefteqn{ 
\Exp(f_n f_m|\ebf)\big|_{n\neq m} 
}\nonumber\\
&= \sum_{j}\sum_{k}\sum_{j'}\sum_{k'} d_j d_{j'} e_{k} e_{k'} 
        \Exp( v_{nj} v_{kj} v_{mj'} v_{k'j'}) .
\end{align}
In this case, the expectation will vanish unless $k=n$ or $k=m$.
When $k=n$, we also need $k'=m$, and 
when $k=m$, we also need $k'=n$ and $j=j'$.
Thus we can write
\begin{align}
\lefteqn{ 
\Exp(f_n f_m|\ebf)\big|_{n\neq m} 
}\nonumber\\
&= e_{n} e_{m} \sum_{j}\sum_{j'} d_j d_{j'} \Exp( v_{nj}^2 v_{mj'}^2 ) 
   + e_{n} e_{m} \sum_{j} d_j^2 \Exp( v_{nj}^2 v_{mj}^2 ) \\
&= 2e_{n} e_{m} \sum_{j} d_j^2 \Exp( v_{nj}^2 v_{mj}^2 ) 
   + e_{n} e_{m} \sum_{j}\sum_{j'\neq j} d_j d_{j'} \Exp( v_{nj}^2 v_{mj'}^2 ) \\
&\stackrel{(a)}{=}
   \frac{2 e_{n} e_{m}}{N(N+2)} \sum_j d_j^2
   + \frac{e_{n} e_{m} (N+1)}{N(N+2)(N-1)} \sum_{j} d_j \sum_{j'\neq j} d_{j'} \\ 
&\stackrel{(b)}{=}
   \frac{N-3}{(N+2)(N-1)} e_{n} e_{m} \frac{1}{N}\sum_j d_j^2 \\
&\textb{\stackrel{(c)}{=} O(1/N)}
\stackrel{N\rightarrow\infty}{=} 0,
\end{align}
where (a) used \eqref{Ev4},
(b) used $\sum_{j'\neq j} d_{j'} = (\sum_{j'}d_{j'})-d_j = -d_j$,
\textb{and (c) used $\frac{1}{N}\sum_j d_j^2=O(1)$ from the definition of $\mc{D}$
and $e_n e_m = O(1)$ from the finite mean and variance of $e_n$.}
This establishes the off-diagonal terms in \eqref{Eff}.
\end{proof}

Lemma~\ref{lem:VDVe} will now be used to establish 
\begin{align}
\Exp(\ebf_2|\ebf_1) 
&\stackrel{N\rightarrow\infty}{=} \zero 
\label{eq:E_ebf2}\\
\Cov(\ebf_2|\ebf_1) 
&\stackrel{N\rightarrow\infty}{=} \varepsilon_2 \Ibf
\label{eq:Cov_ebf2} 
\end{align}
for some $\varepsilon_2>0$.
To simplify the derivation, we first write 
\iftoggle{include_app}{\eqref{ebf2}}{\eqref{ebf2_app}} 
as
\begin{align}
\ebf_2
&= \fbf + \ubf
\quad \text{for} \quad \fbf \defn \Vbf\vec{D}\Vbf\tran\ebf_1
\label{eq:fbf} ,
\end{align}
and recall that $\tr(\vec{D})=0$.
For the mean of $\ebf_2|\ebf_1$, we immediately have that
\begin{align}
\Exp(\ebf_2|\ebf_1)
&= \Exp(\fbf|\ebf_1) + \Exp(\ubf|\ebf_1) 
= \zero 
\end{align}
since $\Exp(\fbf|\ebf_1)=\zero$ due to \eqref{Ef}.
Also, $\Exp(\ubf|\ebf_1)=\zero$ from definition \eqref{ubf}
and $\Exp(\wbf|\ebf_2) = \vec{0}$.
This establishes \eqref{E_ebf2}.

To characterize the covariance of $\ebf_2|\ebf_1$, we write
\begin{align}
\Cov(\ebf_2|\ebf_1) 
&= \Cov(\fbf)
   + \Exp\big[\fbf \ubf\tran\big|\ebf_1\big] 
\nonumber\\&\quad
   + \Exp\big[\ubf \fbf\tran\big|\ebf_1\big] 
   + \Cov(\ubf|\ebf_1) 
\label{eq:Eee2} 
\end{align}
and investigate each term separately.
For the first term in \eqref{Eee2}, 
equation \eqref{Eff} and definition \eqref{fbf} imply that
\begin{align}
\Cov(\fbf|\ebf_1)
&\stackrel{N\rightarrow\infty}{=} \frac{\varepsilon_1}{N} \tr\big[\vec{D}^2\big] \Ibf_N ,
\label{eq:Cov_f} 
\end{align}
for $\varepsilon_1\defn \lim_{n\rightarrow\infty}\frac{1}{N}\sum_{n=1}^N e_{1n}^2$.
For the second and third terms in \eqref{Eee2},  equation \eqref{Ef} and definition \eqref{fbf} imply 
\begin{align}
\Exp\big[\fbf \ubf\tran\big|\ebf_1,\ubf\big]
&\stackrel{N\rightarrow\infty}{=} \zero 
\label{eq:Efu} .
\end{align}
For the last term in \eqref{Eee2}, we can use 
\eqref{Cbf} and
$\Cov(\wbf|\ebf_1)=\Ibf_M/\gamma_w$ to obtain
\begin{align}
\lefteqn{
\Cov(\ubf|\Vbf,\ebf_1)}\nonumber\\
&= \frac{1}{\alpha}\big( \Cbf + \gamma_1\Ibf_N \big)^{-1} \Cbf
\big( \Cbf + \gamma_1\Ibf_N \big)^{-1} \frac{1}{\alpha} \\
&= \frac{1}{\alpha}\Vbf\big( \Lambdabf + \gamma_1\Ibf_N \big)^{-1} \Lambdabf
\big( \Lambdabf + \gamma_1\Ibf_N \big)^{-1} \Vbf\tran \frac{1}{\alpha} \\
&= \Vbf\Sigmabf\Lambdabf^{-1} \Sigmabf \Vbf\tran 
\label{eq:covu|Ve1}
\end{align}
for 
\begin{align}
\Sigmabf 
\defn \frac{1}{\alpha} \big( \Lambdabf + \gamma_1\Ibf_N \big)^{-1} \Lambdabf
= \Ibf_N - \vec{D} 
\label{eq:Sigmabf} .
\end{align}
Then we take the expectation of \eqref{covu|Ve1} over $\Vbf$ to obtain 
\begin{align}
\lefteqn{ [\Cov(\ubf|\ebf_1)]_{n,m} }\nonumber\\
&= \sum_{j=1}^N \frac{(\sigma_j)^2}{\lambda_j} \Exp(v_{nj} v_{mj}|\ebf_1) 
\stackrel{(a)}{=} \delta_{n-m}\frac{1}{N}\sum_{j=1}^N \frac{(\sigma_j)^2}{\lambda_j} ,
\end{align}
where $\sigma_j\defn[\Sigmabf]_{jj}$ and where (a) follows from \eqref{Ev2} and the assumed independence of $\Vbf$ and $\ebf_1$.
Consequently, 
\begin{align}
\Cov(\ubf|\ebf_1)
&= \frac{1}{N}\tr\big[\Sigmabf\Lambdabf^{-1} \Sigmabf\big] \Ibf_N  
\label{eq:Cov_u} .
\end{align}
Combining \eqref{Eee2}--\eqref{Cov_u}, we have
\begin{align}
\Cov(\ebf_2|\ebf_1) &= \varepsilon_2 \Ibf_N 
\end{align}
for
\begin{align}
\varepsilon_2 
\defn \frac{ \varepsilon_1 }{N} \tr\big[(\Ibf_N-\Sigmabf)^2\big] 
          + \frac{1}{N} \tr\big[\Sigmabf\Lambdabf^{-1} \Sigmabf\big] 
\label{eq:Eee2b} .
\end{align}
The expression for $\varepsilon_2$ can be simplified as follows.
\begin{align}
\varepsilon_2 
&= \frac{(\varepsilon_1-1/\gamma_1)}{N} \tr\big[(\Ibf_N-\Sigmabf)^2\big] 
\nonumber\\&\quad
+ \frac{1}{\gamma_1 N} \tr\big[(\Ibf_N-\Sigmabf)^2 
          + \gamma_1 \Sigmabf\Lambdabf^{-1} \Sigmabf\big] \\
&= \frac{(\varepsilon_1-1/\gamma_1)}{N} 
   \sum_{n=1}^N\left(1-\frac{\lambda_n/\alpha}{\lambda_n+\gamma_1}\right)^2
\nonumber\\&\quad
+ \frac{1}{\gamma_1 N} \tr\big[\Ibf_N-2\Sigmabf 
          + \Sigmabf\big(\Ibf_N + \gamma_1 \Lambdabf^{-1}\big) \Sigmabf\big] .
\end{align}
Leveraging \eqref{Sigmabf} to simplify the last term, we get 
\begin{align}
\varepsilon_2 
&= \frac{(\varepsilon_1-1/\gamma_1)}{N} 
   \sum_{n=1}^N\left(\frac{\lambda_n(1-1/\alpha)+\gamma_1}{\lambda_n+\gamma_1}\right)^2
\nonumber\\&\quad
+ \frac{1}{\gamma_1 N} \tr\big[\Ibf_N+(1/\alpha - 2)\Sigmabf \big] \\
&\stackrel{(a)}{=} \frac{(\varepsilon_1-1/\gamma_1)}{N} 
   \sum_{n=1}^N\left(\frac{\lambda_n(1-1/\alpha)+\gamma_1}{\lambda_n+\gamma_1}\right)^2
\nonumber\\&\quad
+ \frac{1}{\gamma_1} \left(\frac{1}{\alpha}-1\right) \\
&\stackrel{(b)}{=} \frac{(\varepsilon_1-1/\gamma_1)}{N} 
   \sum_{n=1}^N\left(\frac{1-\lambda_n/\gamma_2}{1+\lambda_n/\gamma_1}\right)^2
+ \frac{1}{\gamma_2} ,
\end{align}
where (a) used the fact that $\tr(\Sigmabf)=N$ 
and (b) used \eqref{alpha}.

Finally, notice that the elements of $\ebf_2$ come from a sum of the form
\begin{align}
e_{2n}
&= u_n + \sum_{j=1}^N \xi_{nj} e_{j} \text{  for  }\xi_{nj}=[\Vbf\vec{D}\Vbf\tran]_{nj},
\end{align}
where, for any fixed $\ebf_1$, the elements $\{\xi_{nj}\}_{j=1}^N$ are zero mean, $O(1/N)$ variance, and uncorrelated. 
Because $u_n$ are Gaussian, it can be argued using the central limit theorem that the elements of $\ebf_2$ become Gaussian as $N\rightarrow\infty$.
Combining this result with \eqref{E_ebf2}--\eqref{Cov_ebf2}, we have that, 
\textb{given $\ebf_1$, as $N\rightarrow\infty$, the elements of $\ebf_2$ are marginally zero-mean Gaussian and uncorrelated.}

\section{\textb{Experimental Setup}}

\subsection{Multicoil MRI experiments } \label{app:multi}

In this section we detail the experimental setup for the multicoil experiments
\iftoggle{include_app}{%
in Sections~\ref{sec:example}, \textb{\ref{sec:multi_point},} and \ref{sec:multi_line}.
}{%
in Sections~IV-B, IV-C, and IV-D.
}

\subsubsection{Data}
For our multicoil experiments, we used 3T knee and brain data from fastMRI \cite{Zbontar:18}.
For knee training data, we randomly picked $28$ volumes and used the middle $8$ slices from each volume, while for knee testing data we randomly picked $4$ other volumes and used the middle $4$ slices from each.
Only non-fat-suppressed knee data was used.
For brain training data, we randomly picked $28$ volumes and used the bottom $8$ slices from each volume, while for brain testing data we randomly picked $4$ other brain volumes and used the bottom $4$ slides from each.
Only axial T2-weighted brain data was used.
Starting with the raw fastMRI data, we first applied a standard PCA-based coil-compression technique \cite{Buehrer:MRM:07,Zhang:MRM:13} to reduce the number of coils from $C=15$ to $C=8$.
Then we Fourier-transformed each fully-sampled coil measurement to the pixel domain, center-cropped down to size $368\times 368$ so that all images had the same size, and Fourier-transformed back to k-space, yielding fully sampled multicoil k-space measurement vectors $\vec{y}\full\in\Complex^{NC}$ with $N=368^2=135424$ entries.

\subsubsection{Ground-truth extraction}
To extract the ground-truth image $\vec{x}\true$ from $\vec{y}\full$, we first estimated the coil sensitivity maps $\{\vec{s}_c\}_{c=1}^C$ from the central 24$\times$24 region of k-space using ESPIRiT\footnote{We used the default ESPIRiT settings from \href{https://sigpy.readthedocs.io/en/latest/generated/sigpy.mri.app.EspiritCalib.html}{https://sigpy.readthedocs.io/en/latest/generated/sigpy.mri.app.EspiritCalib.html}.} \cite{Uecker:MRM:14}.
We then modeled $\vec{y}\full\approx\vec{A}\full\vec{x}\true$, where according to the definition of $\vec{A}$ we have
\begin{align}
\vec{A}\full
\defn \mat{\vec{F}\Diag(\vec{s}_1)\\[-2mm]\vdots\\\vec{F}\Diag(\vec{s}_C)}
= (\vec{I}_C\otimes\vec{F})\vec{S}
\text{~~for~~} 
\vec{S} \defn \mat{\Diag(\vec{s}_1)\\[-2mm]\vdots\\\Diag(\vec{s}_C)} ,
\end{align}
and we used least-squares to extract the ground-truth images as follows:
\begin{align}
\vec{x}\true 
&\defn (\vec{A}\full\herm\vec{A}\full)^+\vec{A}\full\herm\vec{y}\full \\
&= (\vec{S}\herm\vec{S})^+\vec{S}\herm(\vec{I}_c\otimes\vec{F}\herm)\vec{y}\full \\
&\stackrel{(a)}{=} \vec{S}\herm(\vec{I}_c\otimes\vec{F}\herm)\vec{y}\full \\
&= \vec{A}\full\herm\vec{y}\full
\label{eq:LS},
\end{align}
where (a) holds because ESPIRiT guarantees that, for each index pixel index $n$, the coil maps are either all zero (i.e., $[\vec{s}_c]_n=0~\forall c$) or they have a sum-squared value of one (i.e., $\sum_{c=1}^C |[\vec{s}_c]_n|^2 = 1$).


\subsubsection{Noisy, subsampled, k-space measurements}
To create the noisy subsampled k-space measurements, we started with the fully sampled fastMRI $\vec{y}\full$ from above, applied a sampling mask $\vec{M}$ of acceleration rate $R$, and added circularly symmetric complex-valued WGN $\vec{w}$ to obtain $\vec{y}$.
\textb{The sampling densities that generated the 2D point and 2D line masks were obtained from the \texttt{genPDF} function of the SparseMRI package\footnote{\url{http://people.eecs.berkeley.edu/~mlustig/Software.html}} with the same settings used in the VDAMP code\footnote{\url{https://github.com/charlesmillard/VDAMP}}, except that the 2D line masks used a 1D sampling density while the 2D point masks used a 2D sampling density.}
The variance on the noise was adjusted to reach a desired signal-to-noise ratio (SNR), where $\text{SNR}\defn \|\vec{y}-\vec{w}\|^2/\|\vec{w}\|^2$.
With multicoil data, we used masks with a fully sampled central $24\times 24$ autocalibration region%
\iftoggle{include_app}{, as in \figref{mask}(b)\textb{-(c)},}{}
to facilitate the use of ESPIRiT for coil estimation.

\subsubsection{Algorithm details}
For D-GEC, we used the 2D Haar wavelet transform of depth $D=4$, giving $L=13$ wavelet subbands.
When evaluating $\vec{f}_1$, we use \textb{$150$} CG iterations \textb{in 
\iftoggle{include_app}{\secref{example}}{Section~IV-B} 
and $10$ in Sections
\iftoggle{include_app}{\ref{sec:multi_point} and \ref{sec:multi_line}}{IV-C and IV-D}}.
Also, we use the damping scheme from \cite{Sarkar:ICASSP:21} with a damping factor of $0.3$ and run the D-GEC algorithm for $20$ iterations.  
\textb{For the experiments in 
\iftoggle{include_app}{\secref{example}}{Section IV-B}, 
we used the auto-tuning scheme from \cite{Fletcher:NIPS:17} to adjust $\vec{\gamma}_1$ and $\vec{\gamma}_2$.}

\subsubsection{Denoiser details}
As described in 
\iftoggle{include_app}{\secref{denoising}}{Section IV-A}, 
our corr+corr denoiser was built on bias-free DnCNN \cite{Mohan:ICLR:20}.
For the multicoil experiments, the images were complex-valued and so DnCNN used two input and output channels: one for the real part and one for the imaginary part.
When extending DnCNN to corr+corr, we added a single noise channel, since we assumed that the real and imaginary parts of the noise had the same noise statistics.
Prior to training, each ground-truth image was scaled so that the 98th percentile of its pixel magnitudes equaled $1$.
While training, we used standard deviations $\{1/\sqrt{\gamma_\ell}\}_{\ell=1}^L$ drawn independently from a uniform distribution over a specified interval $[\text{SD}_{\min},\text{SD}_{\max}]$.
Despite the use of a bias-free DNN, we found that it did not work well to train a single denoiser over a very wide range of SDs, and so we trained five different denoisers, each over a different range of subband SDs: $[0,10/255]$, $[10/255,20/255]$, $[20/255,50/255]$, $[50/255,120/255]$, and $[120/255,500/255]$.  
In each case, we used the training procedure described in 
\iftoggle{include_app}{\secref{denoising}}{Section IV-A}, 
with \textb{$\ell_2$} loss, $20$ epochs, a minibatch size of $128$, the Adam optimizer, and a learning rate that \textb{started} at $10^{-3}$ and was reduced by a factor of $2$ at the $8$th, $12$th, $14$th, $16$th, $18$th, and $19$th epochs. 
The denoisers were trained using $64\times 64$ image patches, of which we obtained $645\,792$ from the training images using a stride of $10\times 10$ and standard data-augmentation techniques like rotation and flipping. 
Although we cannot guarantee that the test images will be scaled in the same way, this is not a problem because bias-free DnCNN obeys $\vec{f}_2(\alpha\vec{u},\alpha\vec{N})=\alpha\vec{f}_2(\vec{u},\vec{N})$ for all $\alpha>0$.
It took approximately $24$ hours to train each denoiser on a workstation with a single NVIDIA RTX-A6000 GPU.

\subsection{Single-coil MRI experiments} \label{app:single}

In this section we detail the experimental setup for the single-coil experiments used in 
\iftoggle{include_app}{\secref{single}}{Section IV-E}.

\subsubsection{Data}
For our single-coil experiments, we used MRI images from the Stanford 2D FSE dataset \cite{Ong:ISMRM:18}.
We used the same train/test/validation split from \cite{Metzler:ICASSP:21}: for testing, we used the $10$ images shown in 
\iftoggle{include_app}{\figref{test_images}}{Fig.~3},
for training we used $70$ other images, and for validation we used $8$ remaining images.
All images were real-valued and $352\times 352$.
For each ground-truth image, the fully sampled k-space data was created via $\vec{y}\full=\vec{Fx}\true$ using 2D discrete Fourier transform $\vec{F}$.

\subsubsection{Noisy, subsampled, k-space measurements}
To create the noisy subsampled k-space measurements, we started with the full sampled Stanford $\vec{y}\full$ from above, applied a \textb{2D point} sampling mask $\vec{M}$ of acceleration rate $R$, and added circularly symmetric complex-valued WGN to obtain $\vec{y}$.
The variance on the noise was adjusted to reach an SNR of $45$~dB.
With single-coil data, we do not need a fully sampled central autocalibration region and so we use masks similar to that shown in 
\iftoggle{include_app}{\figref{mask}(a)}{Fig.~1(a)}.

\subsubsection{Algorithm details}
For D-GEC, we used the 2D Haar wavelet transform of depth $D=4$, giving $L=13$ wavelet subbands.  
When evaluating $\vec{f}_1$, we used $10$ CG iterations.
Also, we used the auto-tuning scheme from \cite{Fletcher:NIPS:17} to adjust $\vec{\gamma}_1$ and the damping scheme from \cite{Sarkar:ICASSP:21} with a damping factor of $0.5$.
We ran the D-GEC algorithm for a maximum of $200$ iterations.

\subsubsection{Denoiser details}
As described in 
\iftoggle{include_app}{\secref{denoising}}{Section IV-A}, 
our corr+corr denoiser was built on bias-free DnCNN \cite{Mohan:ICLR:20}.
For the single-coil experiments, the images were real-valued and so the standard DnCNN uses one input and output channel.
When extending that DnCNN to corr+corr, we added a single noise channel.
The training of the denoiser was identical to that used in the multicoil case, described in \appref{multi}.

}{}

\end{document}